\documentclass[
 aip,
 sd,%
 amsmath,amssymb,
preprint,%
nofootinbib]
{revtex4-1}
\usepackage{graphicx}
\usepackage{dcolumn}
\usepackage{bm}
\usepackage{amsthm}
\newtheorem{theorem}{Theorem}
\newtheorem{corollary}{Corollary}
\newtheorem{lemma}{Lemma}

\newcommand{\overbar}[1]{\mkern 1.5mu\overline{\mkern-1.5mu#1\mkern-1.5mu}\mkern 1.5mu}

\begin{document}

\title[Lattice Sum Zeros:3]{Zeros of Lattice Sums: 3. Reduction of  the Generalised Riemann Hypothesis to Specific Geometries}

\author{R.C. McPhedran,\\
School of Physics, University of Sydney,\\
Sydney, NSW Australia 2006.}
  
\begin{abstract} 
 The location of zeros of the basic double sum over the square lattice is studied. This sum can be represented in terms of the product of the Riemann zeta function and the Dirichlet beta function, so that the assertion that all its non-trivial zeros lie on the critical line is a particular case of the Generalised Riemann Hypothesis (GRH). The treatment given here is an extension of that in two previous papers (arxiv:1601.01724, 1602.06330), where it was shown that non trivial zeros of the double sum either lie on the critical line or on lines of unit modulus of an analytic function intersecting the critical line. The extension enables  more specific conclusions to be drawn about the arrangement of zeros of the double sum on the critical line, which are interleaved with zeros of analytic functions, all of which  lie on the critical line. Possible arrangements of zeros are studied, and it is shown that in all identified cases the GRH holds.
\end{abstract}
\maketitle





%


\section{Introduction}
The Riemann Hypothesis (RH) that all non-trivial zeros of the function $\zeta(s)$ lie on the critical line $\Re(s)=\Re(\sigma+i t)=1/2$ is widely regarded as one of the most important and difficult unsolved problems in mathematics\cite{tandhb}. The Generalised Riemann Hypothesis (GRH) that non-trivial zeros of Dirichlet $L$ functions
with integer characters also lie on the critical line has also been widely investigated. The results we present below consist of a number of numerical and analytic investigations of a particular case of the GRH, pertaining to the most important double sum of the Epstein zeta type:
\begin{equation}
S_0(s;\lambda)=\sum_{p_1,p_2}' \frac{1}{(p_1^2+p_2^2 \lambda^2)^s},
\label{et1}
\end{equation}
where the sum over the integers $p_1$ and $p_2$ runs over all integer pairs, apart from $(0,0)$, as indicated by the superscript prime. The quantity $\lambda$ corresponds to the period ratio of the
rectangular lattice, and $s$ is an arbitrary complex number. For $\lambda^2$ an integer, this is an Epstein zeta function, but for $\lambda^2$ non-integer we will refer to it as a lattice sum over the rectangular lattice. Many results connected with lattice  sums of this and more general forms have been collected in the recent book {\em Lattice Sums Then and Now}\cite{lsb}, hereafter denoted {\em LSTN}. For $\lambda=1$, the sum (\ref{et1}) takes a simple form for which the GRH is applicable:
\begin{equation}
S_0(s;1)=4 \zeta(s) L_{-4}(s),
\label{et1a}
\end{equation}
using the notation of Zucker and Robertson \cite{zandr} for Dirichlet $L$ functions.
For $\lambda\neq 1$, in general $S_0(\lambda, s)$ will have non-trivial zeros off the critical line, as was discussed in a previous article\cite{part1} (hereafter referred to as I). The discussion given in a second recent paper \cite{part2} which we build on here
coupled two functions previously studied, denoted as ${\cal T}_-(s)$ \cite{prt} and ${\cal T}_+(s)$ \cite{ki, lagandsuz}, with three new analytic functions, ${\cal K}(1,1;s)$, ${\cal K}(0,0;s)$ and ${\cal K}_\lambda (0,0;s)$. It was shown there that ${\cal K}(0,0;s)$ has zeros both off and on the critical line, while we show here by an asymptotic argument that ${\cal K}_\lambda (0,0;s)$ has all its zeros on the critical line.  In the second paper \cite{part2} combinations of ${\cal K}(1,1;s)$ and ${\cal K}(1,1;1-s)$ were constructed, denoted by  ${\cal U}_{\cal K}(1,1;s)$ and ${\cal V}_{\cal K}(1,1;s)$, and it was shown that all zeros of $S_0(s;1)$ corresponded to
the condition ${\cal U}_{\cal K}(1,1;s)=-1$, or, equivalently, ${\cal V}_{\cal K}(1,1;s)=-1$.

The aim of this paper is to extend the discussion of the second paper\cite{part2} in such a way as to establish a one-to-one correspondence between zeros of $S_0(s;1)$ on the critical line and those of a function ${\cal L}(s)$ previously studied by Lagarias and Suzuki\cite{lagandsuz}, and shown to have all its zeros on the critical line. In fact, one can show that the distribution function of the
zeros of ${\cal L}(s)$ is the same as that of ${\cal T}_-(s)$ and ${\cal T}_+(s)$, with the latter corresponding to any specific value of the function
$\arg \zeta (1+2 i t)$, known to be monotonic in $t$ for $t$ not small\cite{ki}.In this paper, we will concentrate on the case of the square lattice ($\lambda=1$), but will also use results from I\cite{part1}
and II\cite{part2} in the limit as $\lambda\rightarrow 1$, particularly in defining the function ${\cal K}_\lambda (0,0;s)$. While the variety of functions used in the analysis may be at first sight daunting, it demonstrates that the context of double sums and rectangular lattices is richer than that of single sums like that in the Riemann zeta function, and that this greater richness offers extra opportunities for the development of analytic and asymptotic arguments relating to the RH and the GRH. The results in this paper and in its predecessors have been obtained on the basis of extensive numerical investigations, but in the interests of brevity we refer interested readers to papers I and II  for graphical and tabular data. It should be stressed that it is not overly difficult for the expressions  presented below to be implemented in appropriate symbolic software by those interested in their own explorations of the geometric contexts and arguments we describe.

Bogomolny and  Leboeuf\cite{bandl} have  discussed the distribution and  separation of zeros for $S_0( s;1)$, finding that the product form (\ref{et1a}) of this basic sum resulted in
a distribution of zeros with higher probability of smaller gaps than for individual Dirichlet $L$ functions.
Numerical investigations of the distribution and separation of zeros of  more general Epstein zeta functions have been discussed by Hejhal\cite{hejhal1}, and by Bombieri and Hejhal \cite{hejhal2}. Such investigations are difficult for large $t$ even on the most powerful available computers, due to the number of terms required in the most convenient general expansion for the functions (see Section 2) and the degree of cancellation between terms.

Section 2 contains essential results from paper II, both in their form for general $\lambda$ and for $\lambda=1$. These are used in Sections 3 and 4 to prove significant results for double sums, including the division of the complex $s$ plane into extended regions (running from $\sigma=-\infty$ to $\sigma=\infty$), discrete island regions (with bounded variation in $\sigma$ and $t$) and inner island regions within the latter. The most important result is that all zeros of $S_0(s;1)$ not lying on the boundaries between island and inner island regions
must lie on the critical line. Section 5 considers the properties of the zeros and poles of the function ${\cal V}_{\cal K}(1,1;s)$, which within island regions must all lie on the critical line. Sufficient information is provided by the ordering of these zeros and poles to indicate in specific cases that zeros of $S_0(s;1)$ must lie on the critical line. The specific cases considered encompass all those known to occur at this stage.

\section{Rectangular and Square Lattice Sums}
The double sums we consider are, for the rectangular lattice, analytic in the complex variable $s$, and depend on the real parameter $\lambda$. They reduce to sums for the square lattice when $\lambda$ tends to unity. For brevity of notation, we will sometimes omit the parameter $\lambda$ when it takes the value unity. We will also indicate the partial derivative with respect to $\lambda$ by attaching this symbol as a subscript to the function name.

Connected to the double sum  (\ref{et1}) is a  general class of MacDonald function double sums for rectangular lattices:
\begin{equation}
{\cal K}(n,m;s;\lambda)=\pi^n\sum_{p_1,p_2=1}^\infty  \left(\frac{p_2^{s-1/2+n}}{p_1^{s-1/2-n}}\right) K_{s-1/2+m}(2\pi p_1 p_2\lambda).
\label{mac1}
\end{equation}
For $\lambda\geq 1$ and the (possibly complex) number $s$ small in magnitude, such sums converge rapidly, facilitating numerical evaluations. (The sum gives accurate answers
as soon as the argument of the MacDonald function exceeds the modulus of its order by  a factor of 1.3 or so.) The double sums satisfy the following symmetry relation, obtained by interchanging $p_1$ and $p_2$ in the definition (\ref{mac1}):
\begin{equation}
{\cal K}(n,-m;s;\lambda)={\cal K}(n,m;1-s;\lambda).
\label{mac1a}
\end{equation}

The lowest order sum ${\cal K}(0,0;s;\lambda)$ occurs in the representation of $S_0(s;\lambda)$ due to Kober\cite{kober}:
\begin{equation}
\lambda^{s+1/2} \frac{\Gamma(s)}{8\pi^s} S_0(\lambda;s)=\frac{1}{4} \frac{\xi_1(2 s)}{\lambda^{s-1/2}}+\frac{1}{4} \lambda^{s-1/2} \xi_1(2 s-1)+{\cal K}(0,0;s;\frac{1}{\lambda}).
\label{mac2}
\end{equation}
Here $\xi_1(s)$ is the symmetrised zeta function. 
In terms of the Riemann zeta function, (\ref{mac2}) is
\begin{equation}
S_0(s;\lambda)=\frac{2 \zeta (2 s)}{\lambda^{2 s}}+2\sqrt{\pi}\frac{\Gamma(s-1/2) \zeta(2 s-1)}{\Gamma(s)\lambda}+
\frac{8\pi^s}{\Gamma(s) \lambda^{s+1/2}}{\cal K}(0,0;s;\frac{1}{\lambda}).
\label{mac2a}
\end{equation}

A fully symmetrised form of (\ref{mac2}) (symmetric under both $s\rightarrow 1-s$ and $\lambda\rightarrow1/\lambda$)  is:
\begin{equation}
\lambda^{s} \frac{\Gamma(s)}{8\pi^s} S_0(s;\lambda)={\cal T}_+(s;\lambda)+\frac{1}{\sqrt{\lambda}}  {\cal K}(0,0;s;\frac{1}{\lambda}),
\label{mac2s}
\end{equation}
where
\begin{equation}
{\cal T}_+(s;\lambda)=\frac{1}{4}\left[\frac{\xi_1(2 s)}{\lambda^s}+\frac{\xi_1(2 s-1)}{\lambda^{1-s}}\right].
\label{mac2s1}
\end{equation}
Note that ${\cal T}_+( 1-s;\lambda)={\cal T}_+(s;\lambda)$ and $ {\cal K}(0,0;1-s;\lambda)={\cal K}(0,0;s;\lambda)$, so that the left-hand side of equation (\ref{mac2s}) must then be unchanged under replacement of $s$ by $1-s$. The left-hand side is also unchanged under replacement of $\lambda$ by $1/\lambda$, so the same is true for the
sum of the two terms on the right-hand side, although in general it will not be true for them individually. The symmetry relations for $S_0(s;\lambda)$ then are
\begin{equation}
\lambda^{s} \frac{\Gamma(s)}{8\pi^s} S_0(s;\lambda)=\frac{1}{\lambda^{s}} \frac{\Gamma(s)}{8\pi^s} S_0\left(s;\frac{1}{\lambda}\right)=
\lambda^{1-s} \frac{\Gamma(1-s)}{8\pi^{(1-s)}} S_0(1-s;\lambda)=\frac{1}{\lambda^{1-s}} \frac{\Gamma(1-s)}{8\pi^{(1-s)}} S_0\left(1-s;\frac{1}{\lambda}\right).
\label{mac2s2}
\end{equation}
From the equations (\ref{mac2s2}), if $s_0$ is a zero of $S_0(s;\lambda)$ then
\begin{equation}
S_0(s_0;\lambda)=0~ \implies~S_0(s_0;1/\lambda)=0=S_0(1-s_0;1/\lambda)=S_0(1-s_0;\lambda).
\label{mac2s2a}
\end{equation}
Another interesting deduction from (\ref{mac2s}) relates to the derivative of $S_0(\lambda, s_0)$ with respect to $\lambda$:
\begin{eqnarray}
& &\lambda^s S_0(s;\lambda) =\frac{1}{\lambda^s} S_0\left(s;\frac{1}{\lambda}\right)~\implies~\nonumber \\
&& s \lambda^{s-1} S_0(s;\lambda) +\lambda^s \frac{\partial}{\partial \lambda} S_0(s;\lambda) =\frac{-s}{\lambda^{s+1}} S_0\left(s;\frac{1}{\lambda}\right)-
\frac{1}{\lambda^{s+2}}\frac{\partial}{\partial \lambda}  S_0\left(s;\frac{1}{\lambda}\right),
\label{mac2s2b}
\end{eqnarray}
so that
\begin{equation}
\left. \frac{\partial}{\partial \lambda} S_0(s;\lambda)\right|_{\lambda=1}=S_{0,\lambda}(s;1) =-s S_0(s;1).
\label{mac2s2c}
\end{equation}

Combining (\ref{mac2s}) and (\ref{mac2s2}), we arrive at a general symmetry relationship  for ${\cal K}(0,0;s;\lambda)$:
\begin{equation}
{\cal T}_+(s;\lambda)-{\cal T}_+\left(s; \frac{1}{\lambda} \right)=\sqrt{\lambda}{\cal K}(0,0;s;\lambda)-\frac{1}{\sqrt{\lambda}}  {\cal K}\left(0,0;s;\frac{1}{\lambda}\right),
\label{mac2s3}
\end{equation}
or
\begin{eqnarray}
&&\frac{1}{4}\left[\xi_1(2 s)\left(\frac{1}{\lambda^{s}}-\lambda^s\right) +\xi_1(2 s-1)\left(\frac{1}{\lambda^{1-s}}-\lambda^{1-s}\right) \right]=\nonumber\\
&&\sqrt{\lambda}{\cal K}(0,0;s;\lambda)-\frac{1}{\sqrt{\lambda}}  {\cal K}\left(0,0;s;\frac{1}{\lambda}\right).
\label{mac2s4}
\end{eqnarray}
This identity holds for all values of $s$ and $\lambda$. One use of it is to expand about $\lambda=1$, which gives identities for the partial derivatives of 
${\cal K}(0,0;s;\lambda)$ with respect to $\lambda$, evaluated  at $\lambda=1$. The first of these is
\begin{equation}
{\cal L}(s)=\left. -4 \frac{\partial}{\partial \lambda} {\cal T}_+(s;\lambda)\right|_{\lambda=1}=s\xi_1(2 s)+(1-s) \xi_1(2 s-1)=-2 {\cal K}(0,0;s;1)-4 {\cal K}_\lambda(0,0;s;1).
\label{mac2s5}
\end{equation}
All three functions occurring in (\ref{mac2s5}) are even under $s\rightarrow 1-s$.

By analogy to  the equation (\ref{mac2s1}) we define:
\begin{equation}
{\cal T}_-(s;\lambda)=\frac{1}{4}\left[\frac{\xi_1(2 s)}{\lambda^s}-\frac{\xi_1(2 s-1)}{\lambda^{1-s}}\right].
\label{mac2s1a}
\end{equation}
This function is odd under $s\rightarrow 1-s$. For it, the analogue to equation (\ref{mac2s5}) is
\begin{equation}
{\cal L}_-(s)=\left. -4 \frac{\partial}{\partial \lambda} {\cal T}_-(s;\lambda)\right|_{\lambda=1}=s\xi_1(2 s)-(1-s) \xi_1(2 s-1)=-2[ {\cal T}_-(s)+(2 s-1) {\cal T}_+(s)].
\label{mac2s5a}
\end{equation}

It is known\cite{prt,hejhal3,ki,lagandsuz,mcp13} that ${\cal T}_+(s;\lambda)$ and ${\cal T}_-(s;\lambda)$ have all their zeros on the critical line if $\lambda\leq 1$ and $t> 3.9125$. This can be easily seen from the properties of the function
\begin{equation}
{\cal U}(s)=\frac{\xi_1(2 s-1)}{\xi_1(2 s)},
\label{Udef}
\end{equation}
which has modulus smaller than unity to the right of the critical line (where its numerator has its zeros)  and greater than unity to its left (where its denominator has its zeros) for  $t> 3.9125$. By contrast, ${\cal K}(0,0;s)$ has zeros both on the critical line and off it\cite{mcp04}. Also, Lagarias and Suzuki\cite{lagandsuz} have proved that the function we denote by ${\cal L}(s)$ has all its zeros on the critical line. (We may understand this result  since ${\cal L}(s)=0$ if and only if $-1={\cal U}(s) (1-s)/s$. The modulus of the right-hand side is smaller than unity for $\sigma>1/2$, and larger than unity for $\sigma<1/2$.) Similarly, ${\cal L}_-(s)$ has all its zeros on the critical line.  Ki \cite{ki} has proved that the argument of ${\cal U}(1/2+i t)$ is a monotonic decreasing function as $t$ increases, provided that $t>7$. He has also shown that its $t$ derivative is asymptotically $-2 \log (t)$.

In addition to the function ${\cal U}(s)$, we will employ a closely associated function
\begin{equation}
{\cal V}(s)=\frac{{\cal T}_+(s)}{{\cal T}_-(s)}=\frac{1+{\cal U}(s)}{1-{\cal U}(s)}.
\label{Vdef}
\end{equation}
${\cal V}(s)$ is purely imaginary on the critical line and ${\cal U}(s)$ has modulus unity there. The zeros and poles of ${\cal V}(s)$ all lie on the critical line, and all are simple.
The fixed points of the transformation (\ref{Vdef}) are ${\cal V}(s)={\cal U}(s)=\pm i$,
and its normal form is
\begin{equation}
\frac{{\cal V}(s)-i}{{\cal V}(s)+i}=i\left(\frac{{\cal U}(s)-i}{{\cal U}(s)+i}\right).
\label{VUform}
\end{equation}
\section{Properties Related to ${\cal K}(1,1;s;\lambda)$}
The recurrence relations for MacDonald functions give rise to those for the double sums:
\begin{equation}
\frac{\partial}{\partial \lambda} {\cal K}(n,m;s;\lambda)=-[{\cal K}(n+1,m+1;s;\lambda)+{\cal K}(n+1,m-1;s;\lambda)],
\label{mac3}
\end{equation}
and
\begin{equation}
\frac{(m+s-1/2)}{\lambda} {\cal K}(n,m;s;\lambda)=[{\cal K}(n+1,m+1;s;\lambda)-{\cal K}(n+1,m-1;s;\lambda)].
\label{mac4}
\end{equation}
These may be used to construct operators which raise $n$ and lower $m$, or raise $n$ and raise $m$, respectively:
\begin{equation}
-\frac{1}{2}\left[ \frac{\partial}{\partial \lambda}+\frac{(m+s-1/2)}{\lambda}  \right] {\cal K}(n,m;s;\lambda)= {\cal K}(n+1,m-1;s;\lambda),
\label{mac5}
\end{equation}
and
\begin{equation}
-\frac{1}{2}\left[ \frac{\partial}{\partial \lambda}-\frac{(m+s-1/2)}{\lambda}  \right] {\cal K}(n,m;s;\lambda)= {\cal K}(n+1,m+1;s;\lambda).
\label{mac5a}
\end{equation}

From (\ref{mac5a}) we have
\begin{equation}
{\cal K}(1,1;s;\lambda)=-\frac{1}{2}{\cal K}_\lambda (0,0;s;\lambda)+\frac{(s-1/2)}{2 \lambda}{\cal K}(0,0;s;\lambda)
\label{K11def}
\end{equation}
The symmetric and antisymmetric parts of ${\cal K}(1,1;s;\lambda)$ are
\begin{equation}
{\cal K}(1,1;s;\lambda)+{\cal K}(1,1;1-s;\lambda)=-{\cal K}_\lambda (0,0;s;\lambda),
\label{K11sym}
\end{equation}
and
\begin{equation}
{\cal K}(1,1;s;\lambda)-{\cal K}(1,1;1-s;\lambda)=\frac{(s-1/2)}{ \lambda}{\cal K}(0,0;s;\lambda).
\label{K11asym}
\end{equation}

It is  useful to define
\begin{equation}
{\cal V}_{\cal K} (1,1;s;\lambda)=\frac{{\cal K}(1,1;s;\lambda)-{\cal K}(1,1;1-s;\lambda)}{{\cal K}(1,1;s;\lambda)+{\cal K}(1,1;1-s;\lambda)},
\label{etf24}
\end{equation} 
and
\begin{equation}
{\cal U}_{\cal K} (1,1;s;\lambda)=\frac{{\cal K}(1,1;s;\lambda)}{{\cal K}(1,1;1-s;\lambda)}=\left(\frac{1+{\cal V}_{\cal K} (1,1;s;\lambda)}{1-{\cal V}_{\cal K} (1,1;s;\lambda)}\right).
\label{etf24a}
\end{equation}
From equations (\ref{K11sym},\ref{K11asym}),
\begin{equation}
{\cal V}_{\cal K} (1,1;s;\lambda)=\frac{-(s-1/2){\cal K}(0,0;s;\lambda)}{\lambda {\cal K}_\lambda (0,0;s;\lambda)}=\frac{-(s-1/2)}{\lambda \partial {\log\cal K}(0,0;s;\lambda)/\partial \lambda}.
\label{etf25}
\end{equation}
From equations (\ref{etf24a}) and (\ref{etf25}), we have the symmetry relations
\begin{equation}
{\cal U}_{\cal K} (1,1;1-s;\lambda)=\frac{1}{{\cal U}_{\cal K} (1,1;s;\lambda)}, ~{\cal V}_{\cal K} (1,1;1-s;\lambda)=-{\cal V}_{\cal K} (1,1;s;\lambda).
\label{symuv}
\end{equation}

In what follows, we will abbreviate the notation for the sums ${\cal K}$ and their $\lambda$ derivatives by suppressing the entry for the geometric parameter $\lambda$ when it takes the value unity. We will do the same for ${\cal S}_0(s;\lambda)$.

{\bf  Remark 1:} From equation (\ref{etf24}), non-trivial zeros of ${\cal V}_{\cal K} (1,1;s)$ correspond to ${\cal U}_{\cal K} (1,1;s)=1$ and
${\cal K}(0,0;s)=0$; they may lie on or off the critical line. Poles of ${\cal V}_{\cal K} (1,1;s)$ correspond to ${\cal U}_{\cal K} (1,1;s)=-1$ and
${\cal K}_\lambda(0,0;s)=0$; numerical investigations indicate that they lie on the critical line.

{\bf Remark 2:} All zeros and poles of ${\cal V}_{\cal K} (1,1;s)$ lie on lines $|{\cal U}_{\cal K} (1,1;s)|=1$, where ${\cal V}_{\cal K} (1,1;s)$ is pure imaginary. If $\arg[{\cal U}_{\cal K} (1,1;s)]$ varies monotonically along such lines, then zeros and poles of 
${\cal V}_{\cal K} (1,1;s)$ alternate along the lines, and are all of first order.

\section{ ${\cal K}(1,1;s)$ and the Zeros of $S_0(s)$}
From equations (\ref{K11sym}, \ref{K11asym}, \ref{mac2s}) we find:
\begin{equation}
{\cal K}(1,1;s)= s\left[  \frac{\Gamma(s)}{16\pi^s} S_0(s)\right]-(s-1/2) \left[\frac{\xi_1(2s-1)}{4} \right] .
\label{n26}
\end{equation}
Let us define a symmetrised form of  $S_0(s)$:
\begin{equation}
\tilde{S}_0(s)=\frac{\Gamma (s) S_0(s)}{8\pi^s}.
\label{neq4-1}
\end{equation}
Then, putting $\lambda=1$ in equation (\ref{mac2s}), we have
\begin{equation}
\tilde{S}_0(s)={\cal T}_+(s)+{\cal K}(0,0;s).
\label{neq4-2}
\end{equation}
We also have from equation (\ref{mac2s5}):
\begin{equation}
\tilde{S}_0(s)=-2[(s-1/2) {\cal T}_-(s)+{\cal K}_\lambda (0,0;s)].
\label{neq4-3}
\end{equation}
These last  two equations may be re-expressed in a form suitable for computations, giving the MacDonald function sums in terms of
functions readily available in symbolic packages and multiple precision libraries:
\begin{equation}
{\cal K}(0,0;s) =\tilde{S}_0(s)-{\cal T}_+(s),
\label{erep2}
\end{equation}
and
\begin{equation}
{\cal K}_\lambda (0,0;s) =-\frac{1}{2}\tilde{S}_0(s)-\left(s-\frac{1}{2}\right) {\cal T}_-(s).
\label{erep3}
\end{equation}
By contrast with direct summation of the initial form (\ref{mac1}), the forms (\ref{erep2}) and (\ref{erep3}) are well adapted to efficient use when $|s|>>1$.

We have then as a consequence of the equations (\ref{neq4-1},\ref{neq4-2}):
\begin{theorem}
If $\tilde{S}_0(s_0)=0$ for $s_0$ off the critical line, then none of the following can be zero: ${\cal L}(s_0)$, ${\cal T}_+(s_0)$,
${\cal T}_-(s_0)$, ${\cal K} (0,0;s_0)$ and ${\cal K}_\lambda (0,0;s_0)$. Also, if $\tilde{S}_0(s_1)=0$ for $s_1$ on the critical line,
then at least one  of  ${\cal K} (0,0;s_1)$ and ${\cal K}_\lambda (0,0;s_1)$ must be non-zero.
\end{theorem}
\begin{proof}
Since $s_0$ is off the critical line then the following are immediately non-zero since all their zeros lie on the critical line:
${\cal L}(s_0)$, ${\cal T}_+(s_0)$ and ${\cal T}_-(s_0)$. If either ${\cal K} (0,0;s_0)$ or ${\cal K}_\lambda (0,0;s_0)$ is zero as well as 
$\tilde{S}_0(s_0)$, then we have immediately  ${\cal T}_+(s_0)=0$ or
${\cal T}_-(s_0)=0$, respectively- a contradiction.

For the second part of the Theorem,  if $\tilde{S}_0(s_1)=0$ for $s_1$ on the critical line, then at least one of ${\cal K} (0,0;s_1)$ or 
${\cal K}_\lambda (0,0;s_1)$ must be non-zero, since ${\cal T}_+(s)$ and
${\cal T}_-(s)$ have no zeros in common. We can then say that one of ${\cal T}_+(s_1)$ or ${\cal T}_-(s_1)$ must be non-zero.
\end{proof}

Another easily proved result is the following:
\begin{theorem} ${\cal K}(1,1;s)$ and  ${\cal K}(1,1;1-s)$ are not simultaneously zero for $s$ off the critical line. ${\cal K}(0,0;s)$ and  ${\cal K}_{\lambda}(0,0;s)$ are not simultaneously zero for $s$ off the critical line.
\end{theorem}
\begin{proof}
From equations (\ref{K11sym}) and (\ref{K11asym}), if ${\cal K}(1,1;s)$ and  ${\cal K}(1,1;1-s)$ are both zero for some $s$, then so are ${\cal K}(0,0;s)$ and  ${\cal K}_{\lambda}(0,0;s)$.
From (\ref{mac2s5}) we then have ${\cal L}(s)=0$, so that $s$ must lie on the critical line.
\end{proof}
Note that if ${\cal K}(1,1;s)$ and  ${\cal K}(1,1;1-s)$ are equal, then ${\cal K}(0,0;s)$ has to be zero, which can occur for $s$ either on or off the critical line. If ${\cal K}(1,1;s)$ and  $-{\cal K}(1,1;1-s)$ are equal, then ${\cal K}_\lambda (0,0;s)$ has to be zero.

From equations (\ref{etf25}), (\ref{erep2}) and (\ref{erep3}),
\begin{equation}
{\cal V}_{\cal K} (1,1;s)=\frac{\tilde{S}_0(s)-{\cal T}_+(s)}{{\cal T}_-(s)+\tilde{S}_0(s)/(2 (s-1/2))}.
\label{srep6}
\end{equation}
We have also
\begin{equation}
{\cal U}_{\cal K} (1,1;s)=\frac{2 s \tilde{S}_0(s)-(s-1/2) \xi_1(2 s-1)}{2 (1- s) \tilde{S}_0(s)+(s-1/2) \xi_1(2 s)}.
\label{srep4}
\end{equation}

\begin{theorem}
If $S_0(s_0)=0$ then both ${\cal U}_{\cal K} (1,1;s_0)/{\cal U}(s_0)=-1$ and ${\cal V}_{\cal K} (1,1;s_0)/{\cal V}(s_0)=-1$. If 
either of ${\cal U}_{\cal K} (1,1;s_0)/{\cal U}(s_0)=-1$ or ${\cal V}_{\cal K} (1,1;s_0)/{\cal V}(s_0)=-1$ holds, then either ${\cal L}(s_0)=0$, in which case $s_0$ must lie on the critical line, or $S_0(s_0)=0$.
\label{mthm}
\end{theorem}
\begin{proof}
At a zero $s_0$ of $S_0(s)$, we have from (\ref{srep6}) that 
\begin{equation}
{\cal V}_{\cal K} (1,1;s_0)=-\frac{{\cal T}_+(s_0)}{{\cal T}_-(s_0)}=-{\cal V}(s_0).
\label{mthm1a}
\end{equation}
Also, from (\ref{srep4}) and (\ref{Udef}),
\begin{equation}
{\cal U}_{\cal K} (1,1;s_0)=-{\cal U}(s_0). 
\label{mthm2}
\end{equation}

Suppose now that ${\cal U}_{\cal K} (1,1;s_0)/{\cal U}(s_0)=-1$. Then from equations (\ref{etf24}) and (\ref{Vdef}),
\begin{equation}
{\cal V}_{\cal K} (1,1;s_0)=\frac{{\cal U}_{\cal K} (1,1;s_0)-1}{{\cal U}_{\cal K} (1,1;s_0)+1}=-\frac{1+{\cal U}(s_0)}{{1-{\cal U}(s_0)}},
\label{mthm2a}
\end{equation}
so that  ${\cal V}_{\cal K} (1,1;s_0)/{\cal V}(s_0)=-1$. This argument also works in the reverse direction.

Then from (\ref{etf25}) and (\ref{mac2s}),
\begin{equation}
\tilde{S}_0(s)=\frac{{\cal V}_{\cal K} (1,1;s){\cal T}_-(s)+{\cal T}_+(s)}{1-{\cal V}_{\cal K} (1,1;s)/(2 s-1)}=\frac{{\cal T}_-(s)({\cal V}_{\cal K} (1,1;s)+{\cal V}(s))}{1-{\cal V}_{\cal K} (1,1;s)/(2 s-1)}.
\label{mthm3}
\end{equation}
This can also be written as
\begin{equation}
\tilde{S}_0(s)=\frac{2 (2 s-1) {\cal T}_-(s)^2 [{\cal V}_{\cal K} (1,1;s)+{\cal V}(s)]}{{\cal L}(s)-2 {\cal T}_-(s)[{\cal V}_{\cal K} (1,1;s)+{\cal V}(s)]}.
\label{mthm4}
\end{equation}
Hence ${\cal V}_{\cal K} (1,1;s_0)=-{\cal V}(s_0)$ guarantees  $\tilde{S}_0(s)=0$, unless ${\cal L}(s_0)=0$, in which case $\tilde{S}_0(s_0)=- (2 s_0-1) {\cal T}_-(s_0)$.
Note that if ${\cal L}(s_0)=0$, from (\ref {mac2s5a}) we know that $ {\cal T}_+(s_0)\neq 0$ and  $ {\cal T}_-(s_0)\neq 0$, since these two functions have no zeros in common.
Hence, if  ${\cal L}(s_0)=0$ then $\tilde{S}_0(s_0)\neq 0$.
\end{proof}

We next consider that behaviour of  ${\cal U}_{\cal K}(1,1;s)$ in the complex plane for $t$ not small. The theorem which follows shows that this function well away from the critical line
has the opposite behaviour to ${\cal U}(s)$. The latter is smaller than unity in magnitude to the right of the critical line, and larger in magnitude than unity to the left of it. The former has magnitude which increases without bound for $\sigma$ moving well to the right of the critical line, and tends towards zero as $\sigma$ moves well to the left of the critical line.

 \begin{theorem} The function ${\cal U}_{\cal K}(1,1;s)$ has "island" regions symmetric under
 $s\rightarrow 1-\overbar{s}$, defined by boundaries in $\sigma>1/2$ inside which its modulus is less than unity, and outside which it is greater than unity.
  It has  monotonic argument variation around each side of island regions surrounding intervals of the critical line. 
 \label{thmisland}
 \end{theorem}
 \begin{proof}
 We have  from equation (\ref{n26}) the expansion 
 \begin{equation}
 {\cal U}_{\cal K}(1,1;s)=\frac{s {\cal C}(0,1;s)-4\sqrt{\pi}\Gamma (s-1/2) \zeta (2 s-1)/\Gamma (s)}{(1-s) {\cal C}(0,1;s)+4 (s-1/2) \zeta (2 s)}.
 \label{isleq-1}
 \end{equation}
 Using the asymptotic expansion for $\Gamma (s+1/2)/\Gamma (s)$ when $|t|>>1$,
  \begin{equation}
 {\cal U}_{\cal K}(1,1;s)\simeq  \frac{s \zeta (s) L_{-4}(s)-\sqrt{\pi s}(1-1/(8 s)+\ldots)  \zeta (2 s-1)}{(1-s) \zeta(s) L_{-4}(s) + (s-1/2) \zeta (2 s)}.
 \label{isleq-2}
 \end{equation}
 We assume $t$ and $\sigma -1/2$ are sufficiently large so the second term in the numerator is negligible compared with the first, and that in the denominator the series for the product $\zeta (s) L_{-4} (s)$
 and for $\zeta (2 s)$ may be used. We then obtain the following approximation from (\ref{isleq-2}):
 \begin{equation}
 {\cal U}_{\cal K}(1,1;s)\simeq  \frac{2 s (1+1/2^s-\sqrt{\pi/s})}{1- (s-1) 2^{1-s}}.
 \label{isleq-3}
 \end{equation}
 This shows that $|{\cal U}_{\cal K}(1,1;s)|\rightarrow \infty$ as $|s|\rightarrow \infty$ in $\sigma >>1/2$. Hence, regions with $|{\cal U}_{\cal K}(1,1;s)|>1$ in $\sigma>1/2$ must be bounded in their $\sigma$
 range. Since $|{\cal U}_{\cal K}(1,1;\overline{1-s})|=|1/{\cal U}_{\cal K}(1,1;s)|$, regions with $|{\cal U}_{\cal K}(1,1;s)|<1$ in $\sigma<1/2$ must also be bounded in their $\sigma$
 range. If they are also limited in their $t$ range, they form islands symmetric about the critical line, with boundaries given by $|{\cal U}_{\cal K}(1,1;s)|=1$. 
 
 On the island boundary in $\sigma>1/2$, $|{\cal U}_{\cal K}(1,1;s)|$ goes from smaller  than unity in the island to larger than unity outside it, and so by the Cauchy-Riemann equations its argument must increase
 around the boundary in the direction of increasing $t$.  Since the argument of this function is even under $s\rightarrow \overline{1-s}$, it must also increase around the left boundary as $t$ increases. On the critical line within the island region, the argument increases as $t$ decreases.
 \end{proof}
 
 A convenient criterion for deciding whether an interval on the critical line is in an extended region or an island region is that
\begin{equation}
\frac{d}{d t} {\cal U}_{\cal K} (1,1;\frac{1}{2}+i t)<0 ~{\rm in~ an~ island~ region}
\label{srep15}
\end{equation}
and is positive in an extended region (or in an enclave region- see below).

In paper II, figures and tables show the behaviour of the various functions just discussed in the region of three islands, one for $t$ small (around 13), the next for $t$ around 355 and the third for $t$ around 8290. For the second and third cases, the island includes several zeros of $ {\cal U}_{\cal K}(1,1;s)$ in $\sigma>1/2$, as well as poles in $\sigma<1/2$. These poles and zeros serve as the core of {\em inner islands}, defined by boundaries on which $|{\cal U}_{\cal K} (1,1;s_0)/{\cal U}(s_0)|=1$. From Theorem \ref{mthm}, it is evident that any zeros of $S_0(s)$ not on the critical line must lie on the boundaries of inner islands. As well as the inner islands,
the figures in paper II show that the islands also include what we call {\em enclave regions}, having as their core poles of $ {\cal U}_{\cal K}(1,1;s)$ in $\sigma>1/2$, lying closer to the critical line than the zeros of that function. In the presence of enclave regions, the monotonicity of argument variation referred to in Theorem \ref{thmisland} applies to a contour $|{\cal U}_{\cal K}(1,1;s)|=1$ excluding enclaves. Note also that
\begin{equation}
\frac{d}{d t} {\cal U}_{\cal K} (1,1;\frac{1}{2}+i t)>0 ~{\rm in~ an~ enclave~ region},
\label{srep15en}
\end{equation}
just as in an extended region. As well, $|{\cal U}_{\cal K} (1,1;s_0)/{\cal U}(s_0)|>1$ in an enclave region if $\sigma>1/2$; thus any
zero of $S_0(s)$ or ${\cal L}(s)$ in an enclave must lie on the  critical line.

In numerical studies, as reported in II, one can take the proportion of all zeros of $\zeta(s)$ not lying in inner islands as a proxy for the
fraction of zeros assuredly on the critical line, given that this proportion is seen to vary little as $t$ increases.
 The mean value of the fraction for the zeros lying up to $t=8000$ is 0.7266, with the standard deviation being 0.0113.This fraction of course is numerically rather than analytically determined, but it is of interest to compare it with the results established by analytic means for the fraction of zeros proven to lie on the critical line. These have progressed from 1/3 (Levinson\cite{levinson}) to 2/5 (Conrey\cite{conrey1} and 41\% (Bui, Conrey and Young\cite{conrey2}). More recently, Conrey, Iwaniec and Soundararajan\cite{conrey3} have shown that at least 56\% of the  non-trivial zeros in the family of all Dirichlet $L$ functions are simple and lie on the critical line.

In addition to (\ref{srep4}) and (\ref{srep6}), we have for the derivatives with respect to $s$:
\begin{equation}
[ {\cal V}'_{\cal K}(1,1;s)+{\cal V}'(s)]=\frac{2 {\cal U}'_{\cal K}(1,1;s)}{(1+{\cal U}_{\cal K}(1,1;s))^2}+\frac{2 {\cal U}'(s)}{(1-{\cal U}(s))^2}.
\label{srep20}
\end{equation}
Hence, at a zero $s_0$ of $\tilde{S}_0(s)$:
\begin{equation}
[ {\cal V}'_{\cal K}(1,1;s_0)+{\cal V}'(s_0)]=\frac{2 [{\cal U}'_{\cal K}(1,1;s_0)+{\cal U}'(s_0)]}{(1-{\cal U}'(s_0))^2}.
\label{srep21}
\end{equation}
For $s$ on the critical line,
\begin{equation}
\frac{d {\cal U}_{\cal K}(1,1;1/2+ i t)}{d s}={\cal U}_{\cal K}(1,1;1/2+ i t)\frac{d \arg {\cal U}_{\cal K}(1,1;1/2+ i t)}{d t}
\label{srep22}
\end{equation}
and
\begin{equation}
\frac{d {\cal U}(1/2+ i t)}{d s}={\cal U}(1/2+ i t)\frac{d \arg {\cal U}(1/2+ i t)}{d t}.
\label{srep23}
\end{equation}
Hence, for  $s_0$ on the critical line,
\begin{eqnarray}
\frac{d {\cal U}_{\cal K}(1,1;1/2+ i t_0)}{d s}+\frac{d {\cal U}(1/2+ i t_0)}{d s}={\cal U}(1/2+ i t_0)  && \nonumber \\
\left[\frac{d \arg {\cal U}(1,1/2+ i t_0)}{d t}-\frac{d \arg {\cal U}_{\cal K}(1,1;1/2+ i t_0)}{d t}\right]  & &.
\label{srep24}
\end{eqnarray}

{\bf Remark 3:} We know that all zeros of $\tilde{S}_0(s)$ in extended regions  or enclaves lie on the critical line. From equation (\ref{srep24}), we then also know that all these zeros are simple, since  the argument derivatives with respect to $t$ of ${\cal U}_{\cal K}$ and ${\cal U}$ there have opposite signs (with the latter being non-zero).

Returning to the equations (\ref{mthm2}) and (\ref{mthm1a}), we have that zeros of $S_0(s)$ and ${\cal L}(s)$ lie on lines where $|{\cal U}_{\cal K}(1,1;s)|=|{\cal U}(s)|$.
For island regions excluding enclaves, the contours $|{\cal U}_{\cal K}(1,1;s)/{\cal U}(s)|=1$ lie  on separate lines surrounding each of the zeros of ${\cal U}_{\cal K}(1,1;s)$ in $\sigma>1/2$ and poles in $\sigma<1/2$. We will call the regions within islands where  $|{\cal U}_{\cal K}(1,1;s)/{\cal U}(s)|>1$ the {\em outer islands}, adjoining the inner island regions where $|{\cal U}_{\cal K}(1,1;s)/{\cal U}(s)|<1$ and enclave regions.

{\bf Remark 4:} All zeros of $\tilde{S}_0(s)$ not on the critical line must lie on the boundaries between outer and inner islands. 

We define 
\begin{equation}
{\cal F}(s)=\frac{{\cal U}_{\cal K}(1,1;s)}{{\cal U}(s)}=\frac{1+{\cal G}(s)}{1-{\cal G}(s)},~~{\cal G}(s)=\frac{{\cal F}(s)-1}{{\cal F}(s)+1}.
\label{srep25}
\end{equation}
Then the boundaries between inner and outer islands are given by $|{\cal F}(s)|=1$.  In terms of argument derivatives on the critical line, then
the respective conditions for the extended (or enclave) regions, outer island regions and inner island regions are
\begin{equation}
\frac{\partial \arg {\cal U}_{\cal K}(1,1;1/2+ it)}{\partial t}>0,~~\frac{\partial \arg {\cal F}(1/2+ it)}{\partial t}>0,~~\frac{\partial \arg {\cal F}(1/2+ it)}{\partial t}<0.
\label{srep26}
\end{equation}

We can strengthen Remark 3, as follows.

{\bf Remark 5:} $S_0(s)$ has all its zeros on the critical line if and only if lines of argument $\pi$ run from zeros of ${\cal U}_{\cal K}(1,1;s)$ inside inner islands to the critical line without intersecting the boundaries of the inner island (the lines where $|{\cal F}(s)|=1$).

The argument of ${\cal F}(s)$ increases monotonically round the boundary of each inner island, so the condition just enunciated is also equivalent to the requirement that
$|{\cal F}(s)|=1$ and $\arg{\cal F}(s)=\pi$ only hold simultaneously on the critical line.

Number the inner islands with an integer $m$, and let the $m$th inner island segment on the critical line run from $t_l^{(m)}$ up to $t_u^{(m)}$. Let
\begin{equation}
\mu_l^{(m)}=\arg [-{\cal F}(1/2+i t_l^{(m)})],~\mu_u^{(m)}=\arg [-{\cal F}(1/2+i t_u^{(m)})].
\label{fthm1}
\end{equation}
\begin{theorem} The Riemann Hypothesis for $S_0(s)$ holds if and only if $\mu_u^{(m)}<0$, $\mu_l^{(m)}>0$ for all $m$.
\label{lthm}
\end{theorem}
\begin{proof}
If $\mu_u^{(m)}<0$, $\mu_l^{(m)}>0$ for a given $m$, then, as $\arg [-{\cal F}(1/2+i t)]$ increases monotonically as $t$ decreases inside the inner island, there is a point on the critical line between  $t_u^{(m)}$ and  $t_l^{(m)}$ where  $\arg [-{\cal F}(1/2+i t)]=0$. This is then the single point on the boundary $\Gamma_+^{(m)}$ of that part of the inner island in $\sigma\ge 1/2$ where either $S_0(s)=0$ or ${\cal L}(s)=0$. Given this holds for all $m$, all zeros of $S_0(s)$ and ${\cal L}(s)$ in inner islands lie on the critical line. We also know that
all zeros of $S_0(s)$ in enclaves or outside inner islands (i.e. away from the boundaries of inner islands) lie on the critical line, completing this part of the proof.
If the Riemann Hypothesis holds for $S_0(s)$, we know it also holds for ${\cal L}(s)$. Every inner island $m$ has one part on its boundary $\Gamma_+^{(m)}$ where $\arg [-{\cal F}(s)]=0$: this must lie between  $t_u^{(m)}$ and  $t_l^{(m)}$. Given $\arg [-{\cal F}(s)]$ increases as one goes from the former to the latter, then $\mu_u^{(m)}<0$ and  $\mu_l^{(m)}>0$.
\end{proof}
\begin{corollary}
If between every two inner islands there exists at least one point on the critical line where $\arg[-{\cal F}(s)]=0$, i.e. one point on the critical line where either $S_0(s)=0$ or
${\cal L}(s)=0$, then the Riemann Hypothesis holds for $S_0(s)$.
\label{c1lthm}
\end{corollary}
\begin{proof}
Given an inner island $m$, it has at least one point with $t>t_u^{(m)}$ where $\arg [-{\cal F}(s)]=0$, and one with $t<t_l^{(m)}$. Going from the nearest such point above down to
 $t_u^{(m)}$, $\arg [-{\cal F}(s)]$ decreases, and so $\arg[-{\cal F}(1/2+i t_u^{(m)})]<0$. Going from the nearest such point below up to
 $t_l^{(m)}$, $\arg [-{\cal F}(s)]$ increases, and so $\arg[-{\cal F}(1/2+i t_l^{(m)})]>0$. Thus, for all $m$, $\mu_u^{(m)}<0$ and  $\mu_l^{(m)}>0$, so the Riemann Hypothesis holds
 for $S_0(s)$.
\end{proof}
\begin{corollary}
If the Riemann Hypothesis holds for $S_0(s)$, then between every two inner islands there exists at least one point on the critical line where $\arg[-{\cal F}(s)]=0$, i.e. one point on the critical line where either $S_0(s)=0$ or
${\cal L}(s)=0$.
\label{c2lthm}
\end{corollary}
\begin{proof}
If the Riemann Hypothesis holds, then $\mu_u^{(m)}<0$, $\mu_l^{(m)}>0$ for all $m$. Thus, for every inner island $m$, $\mu_u^{(m)}<0$ and $\mu_l^{(m+1)}>0$. Accordingly, there exists
at least one point on the critical line between $t_u^{(m)}$ and $t_l^{(m+1)}$ where $\arg[-{\cal F}(s)]=0$.
\end{proof}

\section{The Geometric Framework for ${\cal V}_{\cal K}(1,1;s)/{\cal V}(s)$}

We now consider the properties of the poles and zeros of ${\cal V}_{\cal K}(1,1;s)/{\cal V}(s)$, and investigate what may be learned concerning the connection between the zeros of ${\cal L}(s)$ and $S_0(s)$.

From the definitions (\ref{Vdef}) and (\ref{etf25}), we have
\begin{equation}
\frac{{\cal V}_{\cal K}(1,1;s)}{{\cal V}(s)}=-\frac{(1-{\cal U}(s))(1-{\cal U}_{\cal K}(1,1;s))}{(1+{\cal U}(s))(1+{\cal U}_{\cal K}(1,1;s))}
=-\left[\frac{1+{\cal U}(s){\cal U}_{\cal K}(1,1;s)-{\cal U}(s)-{\cal U}_{\cal K}(1,1;s)}{1+{\cal U}(s){\cal U}_{\cal K}(1,1;s)+{\cal U}(s)+{\cal U}_{\cal K}(1,1;s)}\right],
\label{nsec1}
\end{equation}
and
\begin{equation}
\frac{{\cal V}_{\cal K}(1,1;s)}{{\cal V}(s)}+1=2\frac{({\cal U}(s)+{\cal U}_{\cal K}(1,1;s))}{(1+{\cal U}(s))(1+{\cal U}_{\cal K}(1,1;s))},~
\frac{{\cal V}_{\cal K}(1,1;s)}{{\cal V}(s)}-1=-2\frac{(1+{\cal U}(s){\cal U}_{\cal K}(1,1;s))}{(1+{\cal U}(s))(1+{\cal U}_{\cal K}(1,1;s))}.
\label{nsec2}
\end{equation}
The inverse relationships to (\ref{nsec1}) and (\ref{nsec2}) are
\begin{equation}
\frac{{\cal U}_{\cal K}(1,1;s)}{{\cal U}(s)}=\frac{({\cal V}(s)+1)({1+\cal V}_{\cal K}(1,1;s))}{({\cal V}(s)-1)(1-{\cal V}_{\cal K}(1,1;s))}=
-\left[\frac{1+{\cal V}(s){\cal V}_{\cal K}(1,1;s)+{\cal V}(s)+{\cal V}_{\cal K}(1,1;s)}{1+{\cal V}(s){\cal V}_{\cal K}(1,1;s)-{\cal V}(s)-{\cal V}_{\cal K}(1,1;s)}\right],
\label{nsec3}
\end{equation}
and
\begin{equation}
\frac{{\cal U}_{\cal K}(1,1;s)}{{\cal U}(s)}+1=+2\frac{({\cal V}(s)+{\cal V}_{\cal K}(1,1;s))}{({\cal V}(s)-1)(1-{\cal V}_{\cal K}(1,1;s))},~
\frac{{\cal U}_{\cal K}(1,1;s)}{{\cal U}(s)}-1=+2\frac{(1+{\cal V}(s){\cal V}_{\cal K}(1,1;s))}{({\cal V}(s)-1)(1-{\cal V}_{\cal K}(1,1;s))}.
\label{nsec4}
\end{equation}
\begin{theorem}
Within an island region (but outside enclaves) the following properties hold:
\begin{enumerate}
\item $|{\cal U}_{\cal K}(1,1;s)|<1$ in $\sigma>1/2$ and $|{\cal U}_{\cal K}(1,1;s)|>1$ in $\sigma<1/2$;
\item All zeros and poles of ${\cal V}_{\cal K}(1,1;s)/{\cal V}(s)$ lie on the critical line;
\item All zeros of $1+{\cal U}(s){\cal U}_{\cal K}(1,1;s)$ lie on the critical line, so that all solutions of ${\cal V}_{\cal K}(1,1;s)/{\cal V}(s)=1$ lie on the critical line;
\item All zeros of ${\cal K}(0,0;s)$ and ${\cal K}_\lambda(0,0;s)$ lie on the critical line;
\end{enumerate}
\label{frthm1}
\end{theorem}
\begin{proof}
Proposition (1) follows from the construction of the island regions, which have the property that $|{\cal U}_{\cal K}(1,1;s)|=1$ on their boundaries, and  in $\sigma>1/2$, $|{\cal U}_{\cal K}(1,1;s)|<1$ properly within them. As remarked in Section II, the two inequalities in Proposition (1) apply everywhere to ${\cal U}(s)$ \cite{ki,lagandsuz}.

Proposition (2)  follows from Proposition (1) and the factored representation for ${\cal V}_{\cal K}(1,1;s)/{\cal V}(s)$ in equation (\ref{nsec1}). Proposition (3) follows from Proposition (1), which leads to the conclusion that within an island region
$|{\cal U}(s){\cal U}_{\cal K}(1,1;s)|=1$ only on the critical line.  Its second part follows from the first part, since if at $s_0$ we have
${\cal V}_{\cal K}(1,1;s_0)={\cal V}(s_0)$, then from equations (\ref{Vdef}) and (\ref{etf24a}), ${\cal U}_{\cal K}(1,1;s_0)=-1/{\cal U}(s_0)$.
Proposition (4) follows from Proposition (2) and equation (\ref{etf25}).
\label{thmnsec}
\end{proof}

\begin{corollary}
Every island region of bounded extent must include a segment of the critical line.
\label{icorollary}
\end{corollary}
\begin{proof}
Consider an island of bounded extent, say ${\cal I}_*$, not including a segment of the critical line. Then it has no zeros of ${\cal K}_\lambda(0,0,s)$ either within it or on its boundary. We may take ${\cal I}_*$ to lie in $\sigma>1/2$, there being its image in $\sigma<1/2$ under
$s\rightarrow 1-\bar{s}$.

On the boundary of ${\cal I}_*$, ${\cal U}_{\cal K}(1,1;s)$ has modulus unity, while ${\cal V}_{\cal K}(1,1;s)$ is purely imaginary.
On the boundary of ${\cal I}_*$ and within it, ${\cal V}_{\cal K}(1,1;s)$ can not have poles. Within ${\cal I}_*$ ${\cal U}_{\cal K}(1,1;s)$
must have a mixture of poles and zeros. Indeed, if it had none of either, it would have to be constant, while if it had only poles or only zeros, its argument on the boundary of ${\cal I}_*$ would vary monotonically through a multiple of $2\pi$, thus passing through
${\cal U}_{\cal K}(1,1;s)=-1$, and requiring ${\cal V}_{\cal K}(1,1;s)$ to have a pole.

As ${\cal U}_{\cal K}(1,1;s)$ is required to have a mixture of poles and zeros within ${\cal I}_*$, there must exist there lines of unit modulus. These must close around individual poles or zeros (while possibly including a segment of the boundary of ${\cal I}_*$).
We then have other internal islands to which we can apply the same argument. As there can only be a finite number of poles or zeros within ${\cal I}_*$, this gives a contradiction.
\end{proof}

In discussing the location of zeros and poles of ${\cal V}_{\cal K}(1,1;s)$ and of ${\cal V}(s)$, and their relationship to the zeros of
${\cal L}(s)$ and of ${\cal S}_0(s)$, we will adopt the notation that zeros of functions are indicated by an abbreviated function symbol.
Thus, if an interval of the critical line, zeros of the three functions ${\cal K}_\lambda(0,0;s)$, ${\cal T}_-(s)$ and ${\cal L}(s)$ occur in that order as $t$ increases, we will denote that as $({\cal K}_\lambda,{\cal T}_-,{\cal L})$ and if a zero of  ${\cal T}_+(s)$ precedes
that of ${\cal K}(0,0;s)$ we will denote that as $({\cal T}_+,{\cal K})$. We will further refine the notation by using round brackets for inner island intervals, square brackets for enclave intervals, and bra-ket brackets for the remaining intervals within islands. 

Using these conventions, the eleven intervals in the island around $t=355$ of paper II are denoted:\\
$<{\cal K}_{\lambda},{\cal T}_-,{\cal L}>$; $(L_{-4},{\cal K})$; $<{\cal T}_+,\zeta>$; $[{\cal K},{\cal T}_-,{\cal L},{\cal K}_\lambda]$;
$<{\cal T}_+, L_{-4}>$; $({\cal T}_-,{\cal K}_\lambda,{\cal L},{\cal K})$;\\
$<{\cal T}_+,\zeta>$;  $({\cal T}_-,{\cal K}_\lambda,{\cal L})$; $<{\cal K},{\cal T}_+,\zeta>$;  $({\cal T}_-,{\cal K}_\lambda,{\cal L})$ and
$<L_{-4}>$.\\
These structures can be understood with the aid of asymptotic analysis, for $t$ not small, with its results confirmed for small $t$ by numerical validation.

\begin{theorem} Inside islands, zeros of ${\cal L}(s)$ occur in the two structures: ${\cal K}_{\lambda},{\cal T}_-,{\cal L}$ and 
${\cal T}_-,{\cal K}_\lambda,{\cal L}$. In enclaves or between islands, zeros of ${\cal L}(s)$ occur in the structure ${\cal T}_-,{\cal L},{\cal K}_\lambda$, while within an enclave a zero of ${\cal K}(0,0;s)$ occurs as well. Consequently, there is a one-to-one mapping between zeros on the critical line of ${\cal K}_\lambda(0,0;s)$, ${\cal T}_-(s)$ and ${\cal L}(s)$.
\label{aboutL}
\end{theorem}
\begin{proof}
If ${\cal T}_-(s)=0$ then $\arg {\cal U}(s)=0$. If ${\cal L}(s)=0$ then from equation (\ref{mac2s5}), if $t>>1$
\begin{equation}
{\cal U}(s)\sim 1-\frac{i}{t},~~\arg {\cal U}(s)\sim -\frac{1}{t}.
\label{aboutL1}
\end{equation}
As $\arg {\cal U}(1/2+i t)$ monotonically decreases as $t$ increases\cite{ki}, then ${\cal T}_-(1/2+i t)=0$ will precede
${\cal L}(1/2+i t)=0$ by an amount which decreases as $t$ increases.

When ${\cal L}(s)=0$, then ${\cal V}_{\cal K}(1,1;s)=-{\cal V}(s)$, while the leading order term for the latter
is 
\begin{equation}
{\cal V}(s)=\frac{1+{\cal U}(s)}{1-{\cal U}(s)}\sim -2 i t.
\label{aboutL2}
\end{equation}
Hence, zeros of ${\cal L}(s)$ are close to poles of ${\cal V}_{\cal K}(1,1;s)$, i.e. to zeros of ${\cal K}_\lambda (0,0;s)$ on the critical line. We next expand argument functions by linear approximations:
\begin{equation}
\arg {\cal  U}(1/2+i t)\sim -\alpha_{\cal U} (t-t_-), \arg {\cal  U}_{\cal K} (1,1;1/2+i t) \sim \pi -\alpha_{{\cal K}\lambda} (t-t_{{\cal K}\lambda}),
\label{aboutL3}
\end{equation}
where the slope $\alpha_{\cal U}$ is always positive in $t>7/2$, while the slope $\alpha_{{\cal K}\lambda}$ is positive in islands but not in enclaves, and is negative in enclaves and between islands. Hence,
\begin{equation}
\frac{{\cal V}_{\cal K}(1,1;1/2+i t)}{{\cal V}(1/2+i t)}=\tan[1/2 \arg {\cal  U}_{\cal K}(1,1;1/2+i t)] \tan[1/2 \arg {\cal U}(1/2+i t)]\sim -\frac{\alpha_{\cal U}}{\alpha_{{\cal K}\lambda}}
\frac{(t-t_-)}{(t-t_{{\cal K}\lambda})}.
\label{aboutL4}
\end{equation}
We now require that ${\cal V}_{\cal K}(1,1;1/2+i t)/{\cal V}(1/2+i t)=-1$ when $t=t_{\cal L}$, and solve to obtain
\begin{equation}
t_{\cal L}-t_-=\left( \frac{\alpha_{{\cal K}\lambda}}{\alpha_{\cal U}}\right) (t_{\cal L}-t_{{\cal K}\lambda}).
\label{aboutL5}
\end{equation}
Hence, if $\alpha_{{\cal K}\lambda}$ is positive, then knowing the left-hand side in  (\ref{aboutL5}) is positive, $t_{\cal L}>t_{{\cal K}\lambda}$
(islands, but not enclaves), while if $\alpha_{{\cal K}\lambda}$ is negative, $t_{\cal L}<t_{{\cal K}\lambda}$ (between islands and  also in enclaves).

We can also be more specific about the order of zeros and poles within islands. Rewrite (\ref{aboutL5}) as
\begin{equation}
t_-=\left(1-\frac{\alpha_{{\cal K}\lambda}}{\alpha_{\cal U}}\right) t_{\cal L}+\frac{\alpha_{{\cal K}\lambda}}{\alpha_{\cal U}}t_{{\cal K}\lambda}.
\label{aboutL6}
\end{equation}
Now, in islands, the intervals between inner islands are those in which $\arg {\cal U}(1/2+i t)$ decreases more rapidly as $t$ increases than does $\arg {\cal U}_{\cal K}(1,1;1/2+i t)$, while inside inner islands the latter decreases more rapidly than the former.
Hence, between inner islands (\ref{aboutL6}) gives $t_{{\cal K}\lambda}<t_-<t_{\cal L}$, i.e. ${\cal K}_{\lambda},{\cal T}_-,{\cal L}$.
Inside inner islands, $t_-<t_{{\cal K}\lambda}<t_{\cal L}$, i.e. ${\cal T}_-,{\cal K}_\lambda,{\cal L}$.

Next, consider the case of an enclave. The boundary of the enclave in $\sigma\ge1/2$ is given by $|{\cal U}_{\cal K}(1,1;s)|=1$
or $\arg {\cal V}_{\cal K}(1,1;s)=\pm \pi/2$. It contains a single pole within it, so that $\arg {\cal U}_{\cal K}(1,1;s)$ takes all values from $0$ to $2\pi$ round the boundary. Hence, it has a point corresponding to ${\cal L}(s)=0$ on its boundary, and in fact on $\sigma=1/2$. Thus, it contains the triplet ${\cal T}_-,{\cal L},{\cal K}_\lambda$. The point on the enclave boundary corresponding
to $\arg  {\cal U}_{\cal K}(1,1;s)=0$ corresponds to ${\cal K}(0,0;s)=0$. Given the point where ${\cal K}(0,0;s)=0$ is close to that where ${\cal K}_\lambda=0$ (closeness being defined by comparison with $\delta t$ for the enclave), then the whole of the enclave boundary in $\sigma>1/2$ and the greater part of the critical line segment in the enclave will have a common argument for ${\cal V}_{\cal K}(1,1;s)$.

The one-to-one correspondence between zeros on the critical line of ${\cal K}_\lambda(0,0;s)$, ${\cal T}_-(s)$ and ${\cal L}(s)$ referred to in the theorem statement  is provided by their joint occurrence in sets of three (or triples).
\end{proof}
\begin{figure}[htb]
\includegraphics[width=2.5in]{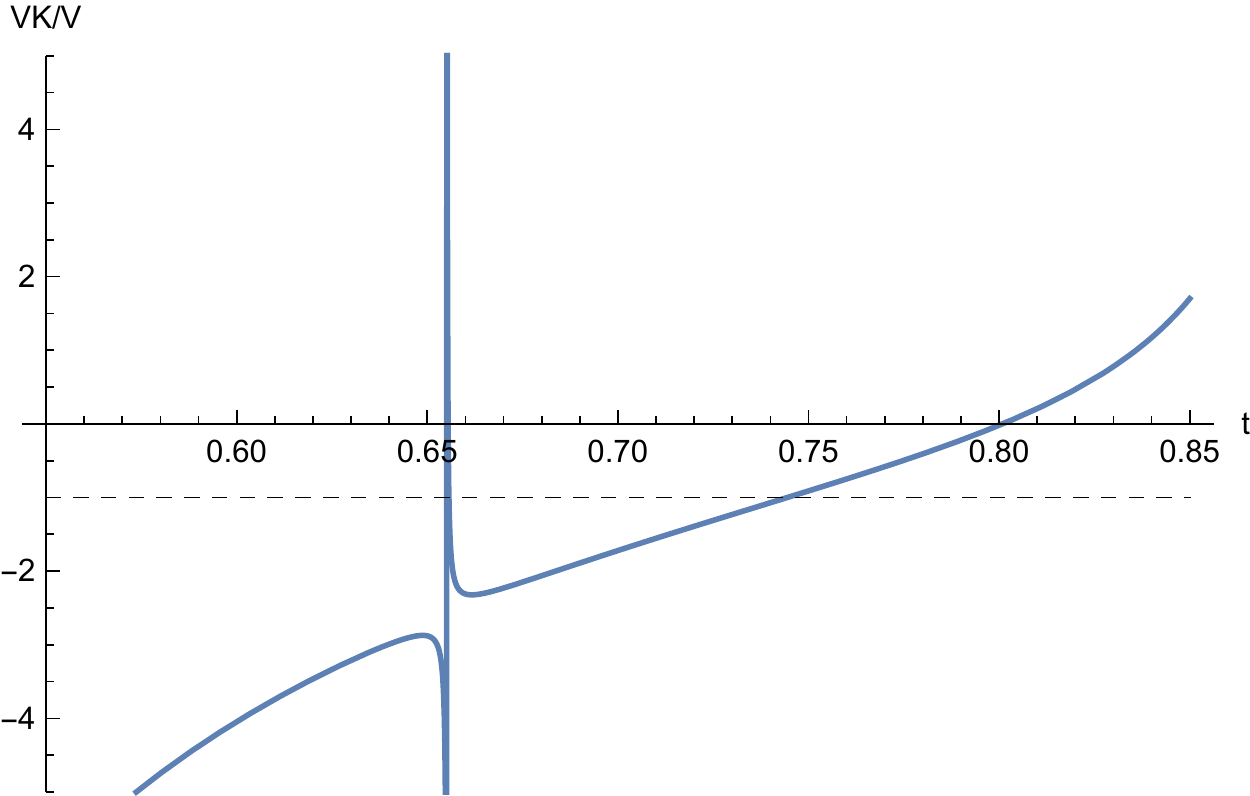}~~\includegraphics[width=2.5in]{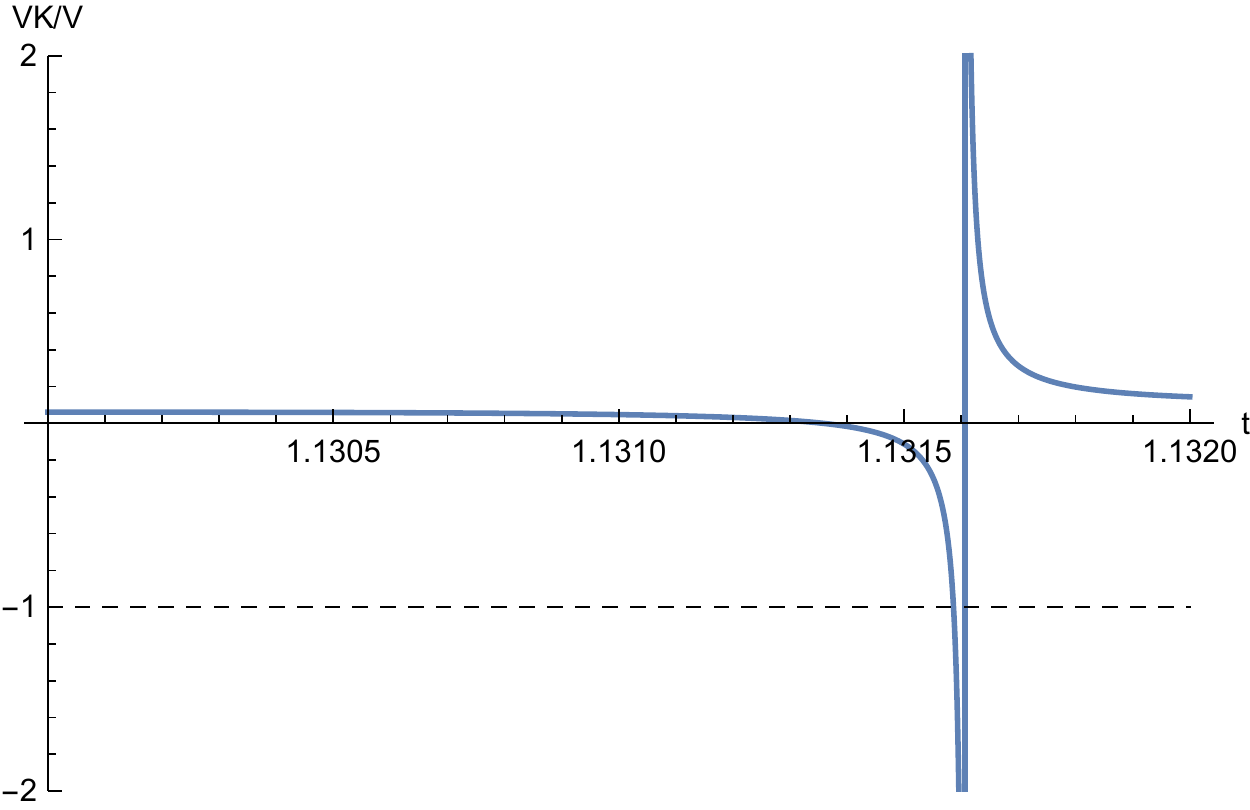}\\
\includegraphics[width=2.5in]{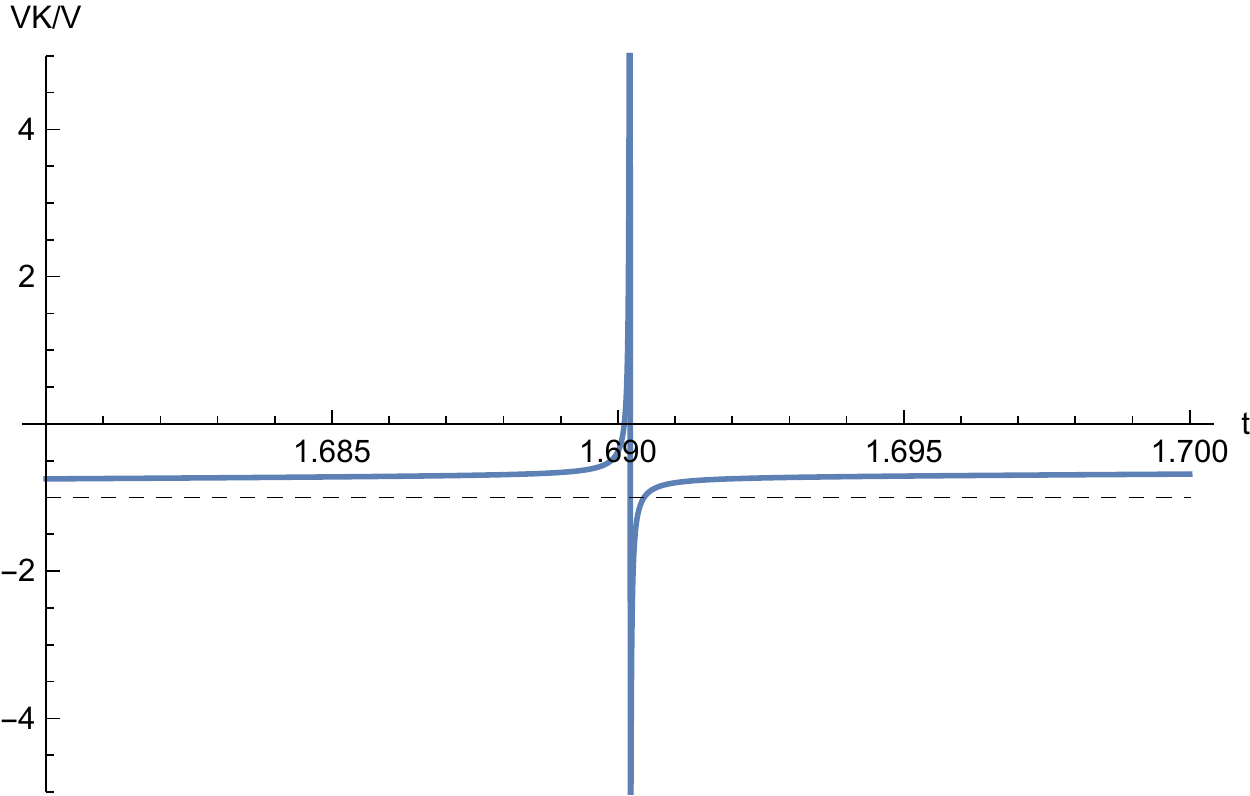}
\caption{Plots of ${\cal V}_{\cal K}(1,1;1/2+i(355+t))/{\cal V}(1/2+i(355+t))$ versus $t$ for thee intervals in an island, showing
details of triple zeros: top left:${\cal K}_{\lambda},{\cal T}_-,{\cal L}$; top right:  ${\cal T}_-,{\cal L},{\cal K}_\lambda$;
bottom:${\cal T}_-,{\cal K}_\lambda,{\cal L}$.}
\label{figtrips1}
\end{figure}

Three examples of triples in the region of the critical line following $t=355$ are given in Fig. \ref{figtrips1}. Of the five occurrences of triples in this island, those in Fig. \ref{figtrips1}(a),(b) occur once, while that in Fig. \ref{figtrips1}(c) occurs three times.

\begin{theorem}
Let $N_Z({\cal K}_\lambda)$, $N_Z({\cal K})$, $N_Z({\cal U}_{\cal K})$ , $N_Z({\cal U})$,$N_Z({\cal T_+})$  and $N_Z({\cal T_-})$ denote the numbers of zeros
of respectively ${\cal K}_\lambda (0,0;s)$, ${\cal K} (0,0;s)$, ${\cal U}_{\cal K}(1,1;s)$, ${\cal U}(s)$, ${\cal T}_+(s)$  and ${\cal T}_-(s)$ in $\sigma \ge 1/2$ in an island region having $N_{\cal E}$ enclaves. Then
\begin{equation}
N_Z({\cal K}_\lambda)=N_Z({\cal K})=N_Z({\cal U}_{\cal K})+N_{\cal E}=N_Z({\cal T}_-),
\label{islzct1}
\end{equation}
and 
\begin{equation}
N_Z({\cal U})=\Big\lfloor  \frac{N_Z({\cal T_+})+N_Z({\cal T_-})}{2}  \Big\rfloor .
\label{islzct1a}
\end{equation}
\label{islzct}
\end{theorem}
\begin{proof}
Let $\Gamma_{1+}$ and $\Gamma_{2+}$ denote the boundary of the contour bounding the island in $\sigma \ge 1/2$, including a segment of the critical line and excluding all enclave regions, and the same contour in $\sigma \ge 1/2$ and including all enclave regions. We  apply the Argument Principle to ${\cal U}_{\cal K}(1,1;s)$ on $\Gamma_{1+}$, giving a change of argument round this contour of $2\pi N_Z({\cal U}_{\cal K})$, so that $N_Z({\cal U}_{\cal K})$ also gives the 
multiplicity of values of the monotonic argument function round the contour. It then follows that zeros of ${\cal K}(0,0;s)$ and of 
${\cal K}_\lambda (0,0;s)$ alternate round $\Gamma_{1+}$, and that the number of each is $N_Z({\cal U}_{\cal K})$. This is 
the number of zeros  on the critical line but excluding the enclaves.    Adding in the common number of zeros  from enclaves gives the equation (\ref{islzct1}).

Note that since zeros of ${\cal K}(0,0;s)$ and of 
${\cal K}_\lambda (0,0;s)$ alternate round $\Gamma_{1+}$, and since the latter only lie on the critical line, there can be at most one
zero of ${\cal K}(0,0;s)$ off the critical line. If there is one such, then the critical line segment pertaining to the island will have a zero of ${\cal K}_\lambda (0,0;s)$ towards either end.

All zeros of 
${\cal U}(s)$ inside the contours $\Gamma_{1+}$ and $\Gamma_{2+}$  lie outside the inner islands and enclaves, and by the Argument Principle the change of argument of ${\cal U}(s)$ round either contour is $2 \pi N_Z({\cal U})$. All zeros and poles of ${\cal V}(s)$ lie on the critical line, as indeed do the points where ${\cal U}(s)$  and ${\cal V}(s)$ both take the values $i$ or $-i$. Now the
zeros of ${\cal T}_+(s)$ correspond to ${\cal U}(s)=-1$, while zeros of ${\cal T}_-(s)$ correspond to ${\cal U}(s)=+1$, and both must occur within a segment of the critical line where the change of argument of ${\cal U}(s)$  is $2\pi$. This leads to the expression 
(\ref{islzct1a}) for $N_Z({\cal U})$.
\end{proof}

An example of an island having no zero of ${\cal U}(s)$ within it may be found near $t=116$. The island contains zeros of
${\cal U}_{\cal K}(1,1;s)$, ${\cal T}_-(s)$ and ${\cal L}(s)$. It does not contain a zero of ${\cal T}_+(s)$.

\begin{corollary}
Any island having no enclaves has all its zeros of $S_0(s)$ on the critical line.
\label{corollencl0}
\end{corollary}
\begin{proof}
We have an island having an arbitrary number of zeros $N_Z({\cal U})$, but with $N_{\cal E}=0$. The island begins and ends with intervals of the critical line not corresponding to an inner island, since the island begins and ends where $\arg {\cal U}_{\cal K}(1,1;1/2+i t)$ changes from increasing with $t$  to decreasing. For an inner island, it must decrease more rapidly than
$\arg {\cal U}(1/2+i t)$. We need only prove that inner islands not having ${\cal L}(s)=0$ on the critical line within them correspond to
$S_0(1/2+i t)=0$ for $t$ on the critical line. Such an inner island must have on either side an interval of the critical line with a triple of the form 
${\cal K}_{\lambda},{\cal T}_-,{\cal L}$, and the two intervals must sandwich either ${\cal K} {\cal T}_+$ or ${\cal T}_+{\cal K}$.
The order of these zeros in fact constrains the form of the graph of ${\cal V}_{\cal K}(1,1;s)/{\cal V}(s)$ with $s=1/2+i t$. A schematic
of the unique forms of each possibility is given in Fig. \ref{figencl0}. (In calculating such schematic graphs, the zeros and poles are specified in the desired order, while the intersection points at the level -1 follow by continuity of the graphs.)  In the first case, the zero of $S_0(s)$ occurs on the critical line
to the right of the pole corresponding to ${\cal T}_+(s)$ and the zero corresponding to ${\cal K}(0,0;s)$. In the second case, the zero of $S_0(s)$ occurs on the critical line just after the point where ${\cal L}(s)=0$, before the zero corresponding to ${\cal K}(0,0;s)$
and the pole corresponding to ${\cal T}_+(s)$. In both cases, the existence of a zero of $S_0(s)$ on the critical line is a consequence
of the Intermediate Value Theorem. In the first case, we have two first-order poles of ${\cal V}_{\cal K}(1,1;s)/{\cal V}(s)$ with only a single first-order zero between them; in the second case,  the function has passed from above -1 to below it, and is constrained to have a value tending towards positive infinity.
\end{proof}

\begin{figure}[htb]
\includegraphics[width=2.5in]{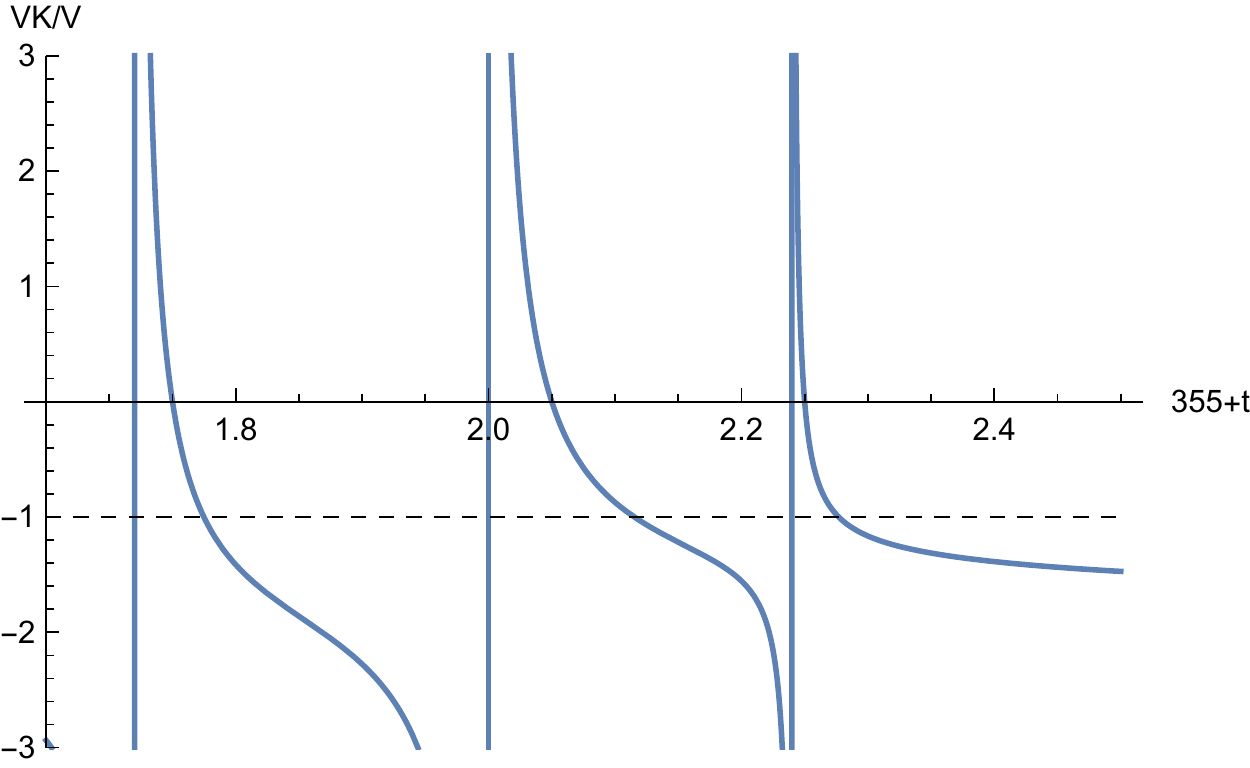}~~\includegraphics[width=2.5in]{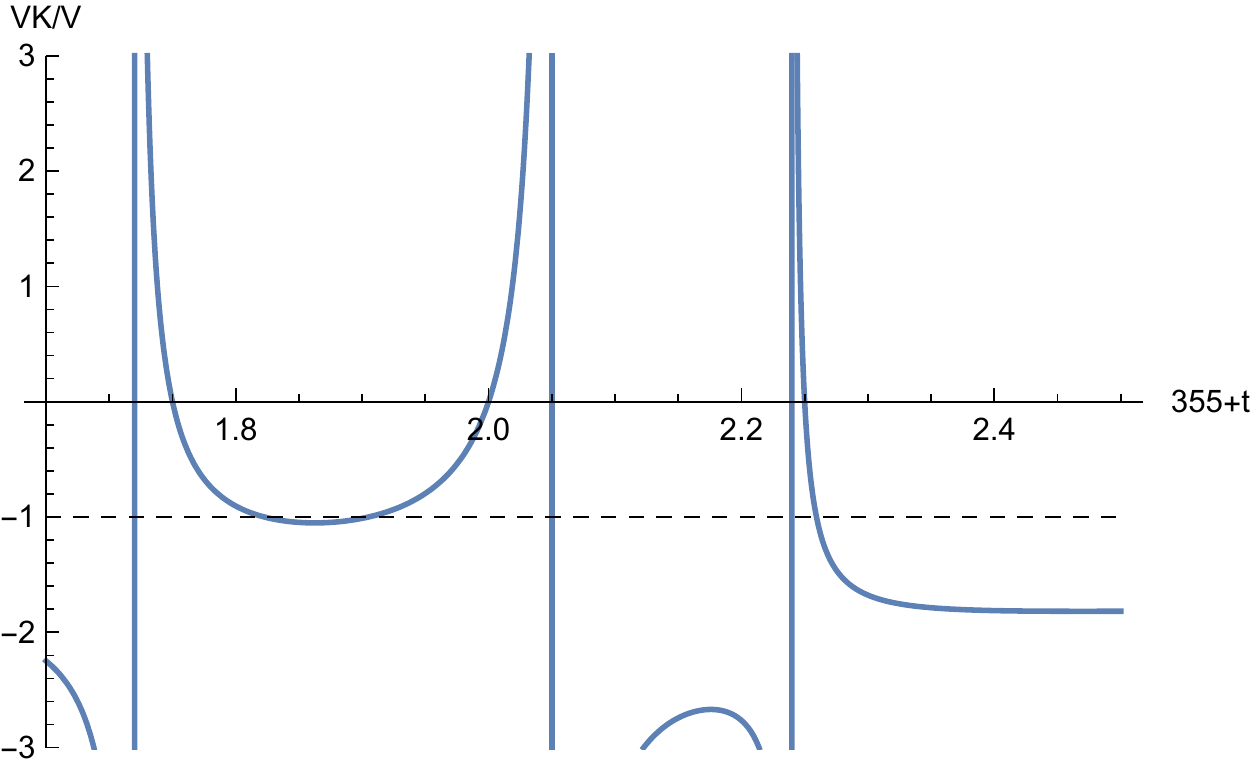}
\caption{Schematic plots of ${\cal V}_{\cal K}(1,1;1/2+i(355+t))/{\cal V}(1/2+i(355+t))$ versus $t$ , showing
the two cases: at left, ${\cal K}_{\lambda},{\cal T}_-,{\cal L}:{\cal T}_+,{\cal K}: {\cal K}_{\lambda},{\cal T}_-,{\cal L}$; at right,  
${\cal K}_{\lambda},{\cal T}_-,{\cal L}:{\cal K},{\cal T}_+: {\cal K}_{\lambda},{\cal T}_-,{\cal L}$. The dashed line intersects the continuous curves at three points: the first and third denote zeros of ${\cal L}(s)$, and the second denotes a zero of $S_0(s)$.}
\label{figencl0}
\end{figure}

\begin{corollary}
Any island having a single enclave within it has all its zeros of $S_0(s)$ on the critical line.
\label{corollencl1}
\end{corollary}
\begin{proof}
We  consider  the case of an enclave  with the structure ${\cal K},{\cal T}_-,{\cal L},{\cal K}_\lambda$ 
occurring first, followed by the two ${\cal T}_+,{\cal T}_-$ alternatives
and ${\cal K}_{\lambda},{\cal T}_-,{\cal L}$. Schematic graphs for the two alternatives are given in Fig. \ref{figencl1a}. The zero of
$S_0(s)$ occurs on the critical line just to the left of the singularity corresponding to ${\cal T}_+(s)=0$ in the first alternative, and just to its right in the second. In the first case, the zero of $S_0(s)$ is constrained in its location by being above a pole coming from negative infinity and below a zero. In the second case, the function first passes through a zero going negative and second has a pole at negative infinity, constraining it to pass through the value -1. 

We  next consider  the case of a structure ${\cal K},{\cal T}_-,{\cal L},{\cal K}_\lambda$ 
occurring first, followed by the two ${\cal K},{\cal T}_+$ alternatives
and ending with an enclave  ${\cal K},{\cal T}_-,{\cal L},{\cal K}_\lambda$ .
Schematic graphs for the two alternatives are given in Fig. \ref{figencl1b}. The case at left is clearcut: the two zeros of $S_0(s)$
occur each associated with a zero of ${\cal L}(s)$, and the existence of the former on the critical line is guaranteed by that of the latter. For the case at right, the two zeros of $S_0(s)$ lie between two zeros of ${\cal K}(0,0,s)$ and are not associated with zeros of ${\cal L}(s)$, so a different argument is needed from the case at left. In fact, the diagram at right requires either two occurrences of ${\cal V}_{\cal K}(1,1;1/2+i t)/{\cal V}(1/2+i t)=-1$ to occur on the critical line if the dashed line passes above the minimum, or no occurrences if it passes below. This means that either two or no zeros of $S_0(s)$ occur on the critical line in the interval in question.
However, we know that in the interval between the enclave and the preceding interval containing a zero of ${\cal L}(s)$ there must be
both an interval not an inner island, and an inner island (since each inner island can only accommodate one zero of $S_0(s)$). As well, the interval which is not an inner island can hold no zeros of $S_0 (s)$ off the critical line. Thus, this second argument shows that in fact there is only the possibility of one or no zeros of $S_0 (s)$ off the critical line in the region of interest. Combining these two facts, we see that in fact there can be no zeros off the critical line in the interval, and two zeros on it.
  \end{proof}

\begin{figure}[htb]
\includegraphics[width=2.5in]{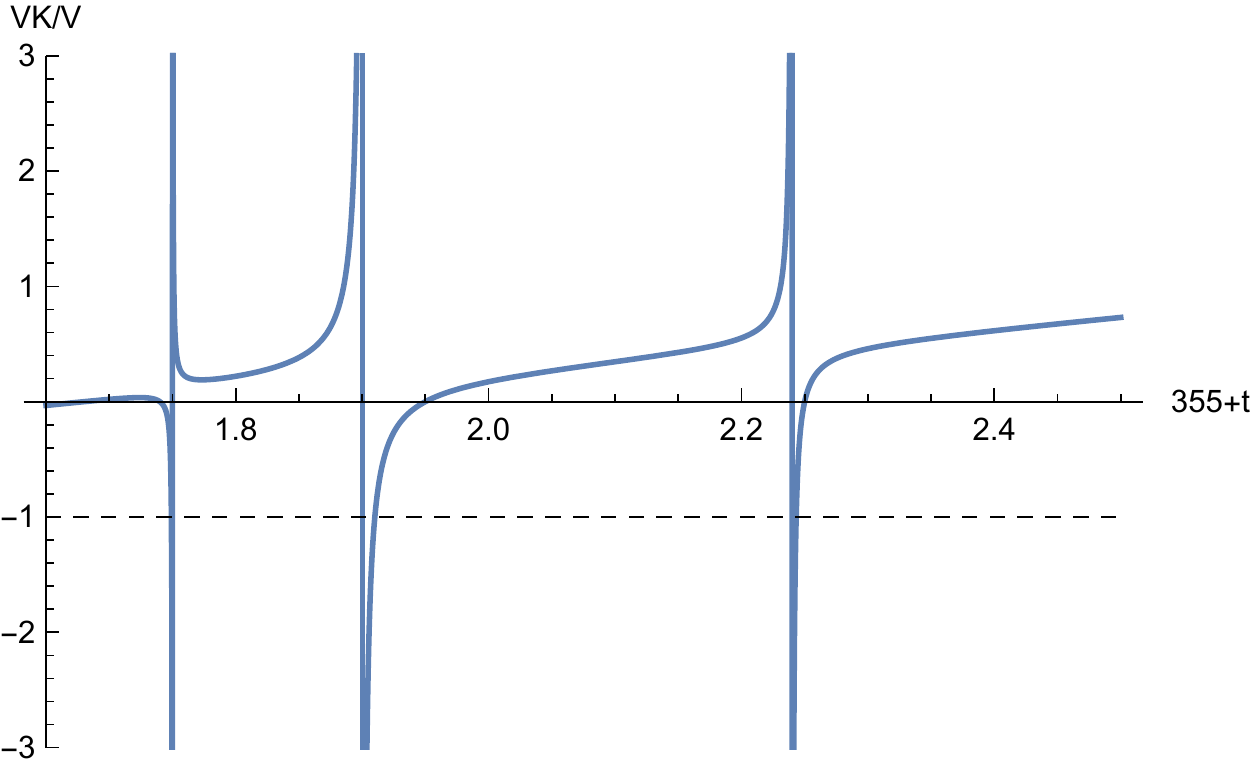}~~\includegraphics[width=2.5in]{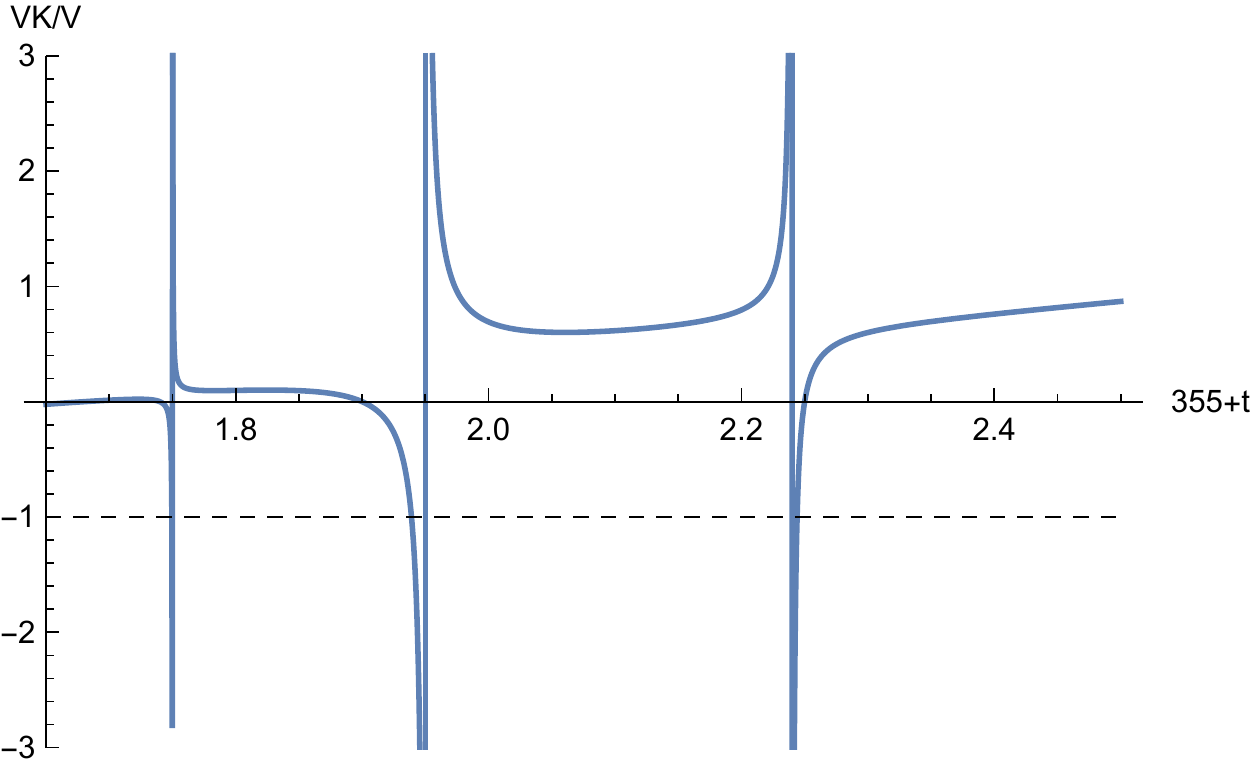}
\caption{Schematic plots of ${\cal V}_{\cal K}(1,1;1/2+i(355+t))/{\cal V}(1/2+i(355+t))$ versus $t$ , showing
the two cases beginning with an enclave: at left, ${\cal K},{\cal T}_-,{\cal L},{\cal K}_\lambda:{\cal T}_+,{\cal T}_-: {\cal K}_{\lambda},{\cal L}_,{\cal K}$; at right,  
${\cal K},{\cal T}_-,{\cal L},{\cal K}_\lambda:{\cal T}_-,{\cal T}_+: {\cal K}_{\lambda},{\cal L}_,{\cal K}$. The dashed line intersects the continuous curves at three points: the first and third denote zeros of ${\cal L}(s)$, and the second denotes a zero of $S_0(s)$.}
\label{figencl1a}
\end{figure}

\begin{figure}[htb]
\includegraphics[width=2.5in]{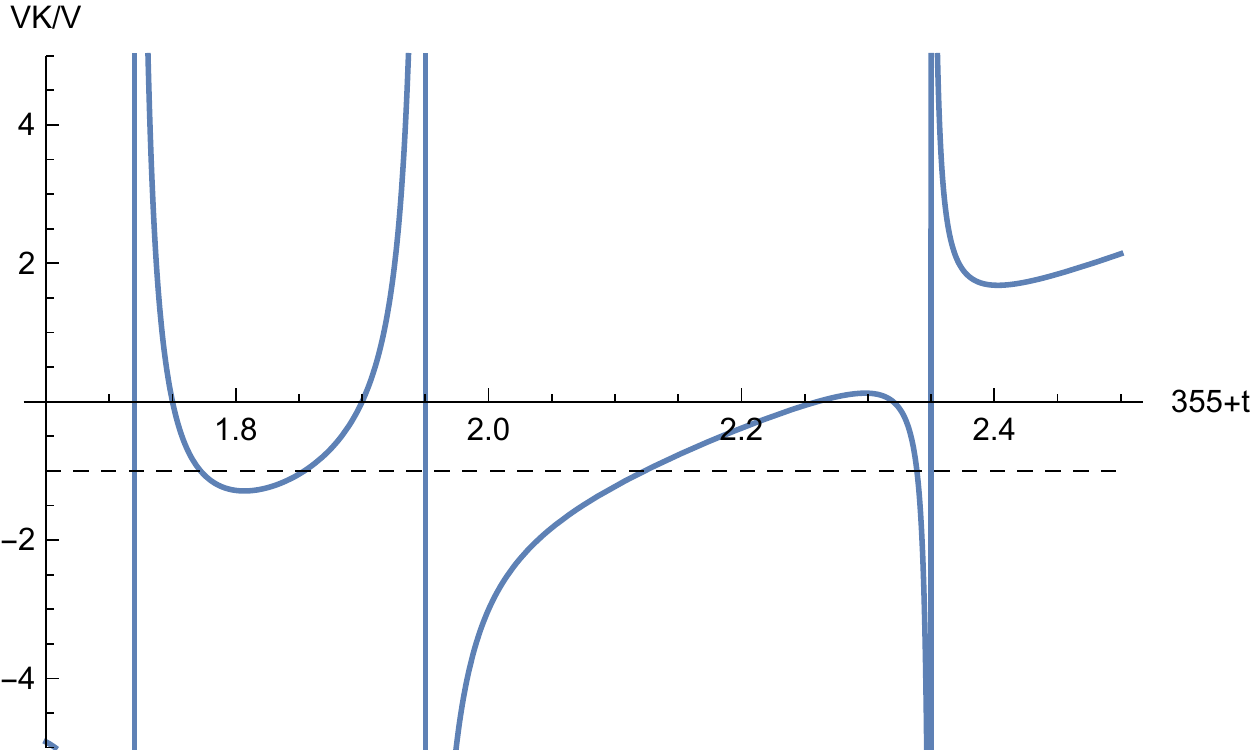}~~\includegraphics[width=2.5in]{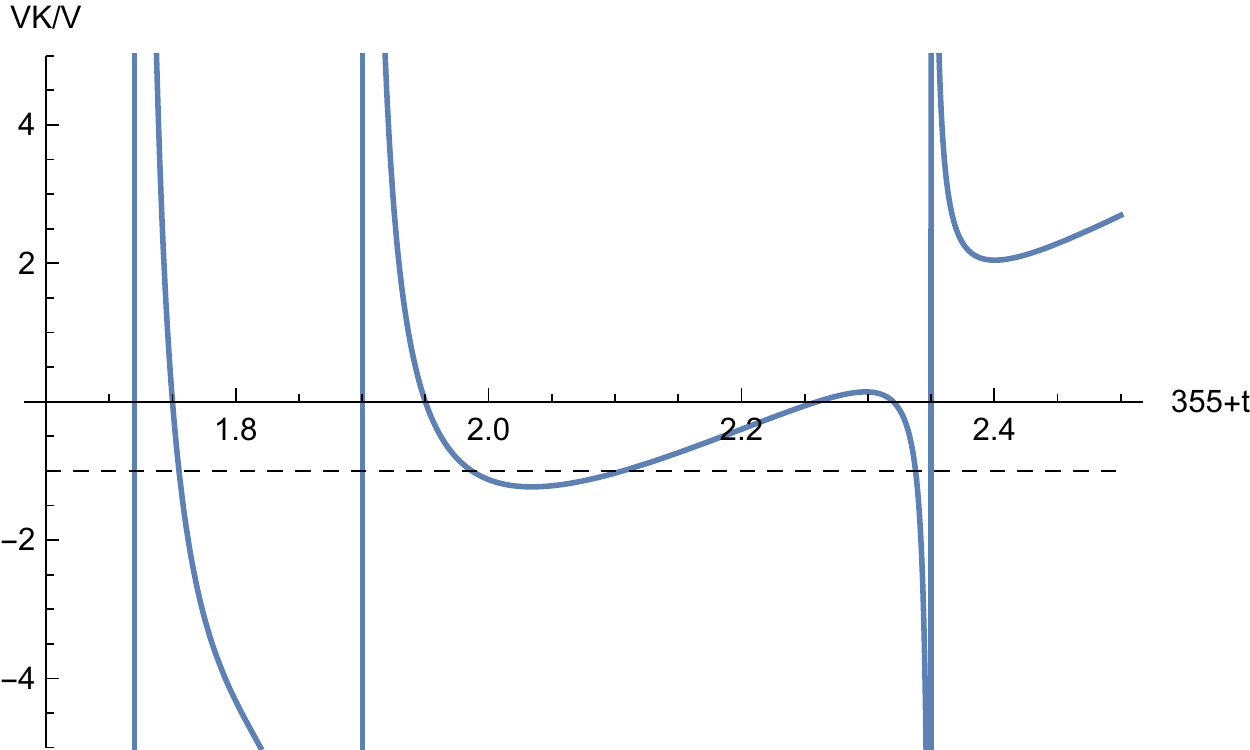}
\caption{Schematic plots of ${\cal V}_{\cal K}(1,1;1/2+i(355+t))/{\cal V}(1/2+i(355+t))$ versus $t$ , showing
the two cases ending with an enclave: at left, $ {\cal K}_{\lambda},{\cal T}_-,{\cal L} :{\cal K} ,{\cal T}_+ : {\cal K},{\cal T}_-,{\cal L},{\cal K}_\lambda$; at right,  $ {\cal K}_{\lambda},{\cal T}_-,{\cal L} :{\cal T}_+,{\cal K} : {\cal K},{\cal T}_-,{\cal L},{\cal K}_\lambda$
. The dashed line intersects the continuous curves at four points: the first and fourth denote zeros of ${\cal L}(s)$, and the second and third denote a zero of $S_0(s)$.}
\label{figencl1b}
\end{figure}

\begin{corollary}
Any island having two or more  enclaves within it has all its zeros of $S_0(s)$ on the critical line.
\label{corollencl2}
\end{corollary}
\begin{proof}
Given Corollaries \ref{corollencl0} and \ref{corollencl1}, the only additional situation we need to consider is that of an interval
bounded by two enclaves. Schematic graphs for the two alternatives  for the interval between the enclaves are given in Fig. \ref{figencl2}. Note that the second enclave must have the zero of ${\cal K}(0,0;s)$ at its end rather than at its beginning, to ensure graphical consistency. In both cases, the single zero of $S_0(s)$ lies in the interval between two poles which includes a single zero of
${\cal V}_{\cal K}(1,1;s)/{\cal V}(s)$, and thus must be located on the critical line.
\end{proof}

\begin{figure}[htb]
\includegraphics[width=2.5in]{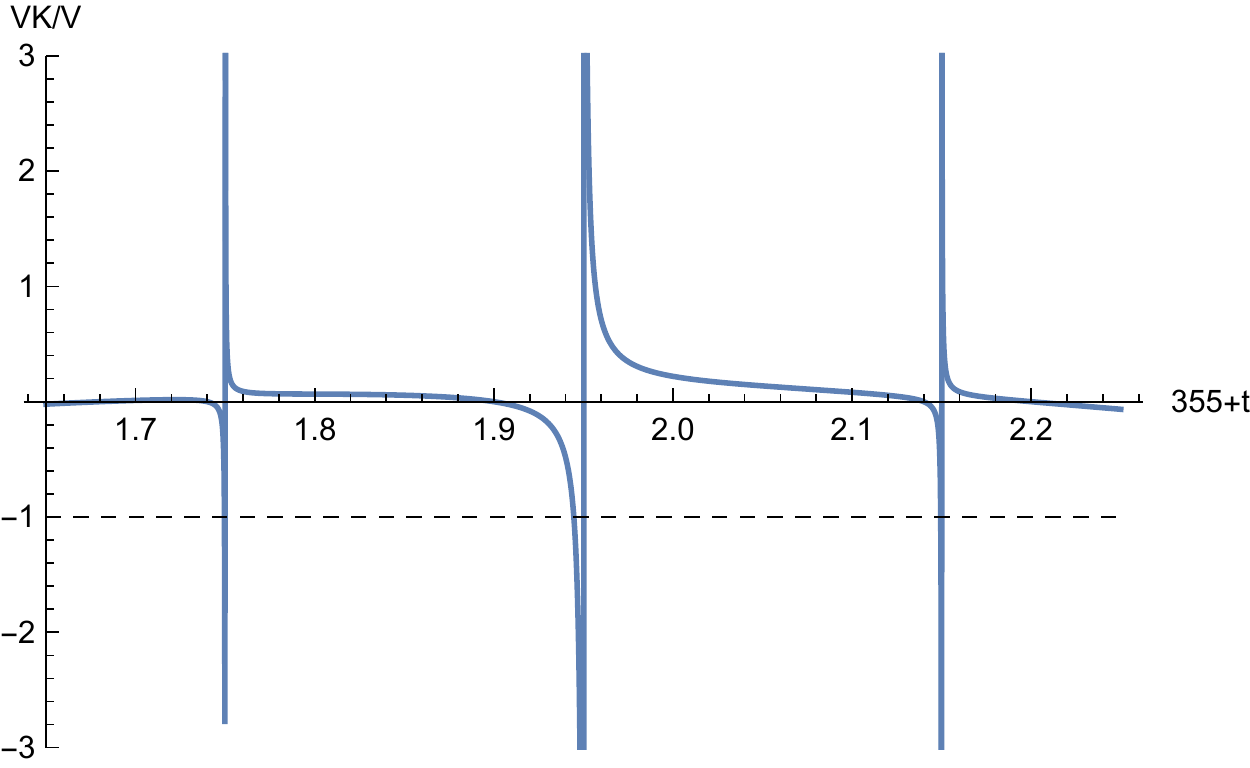}~~\includegraphics[width=2.5in]{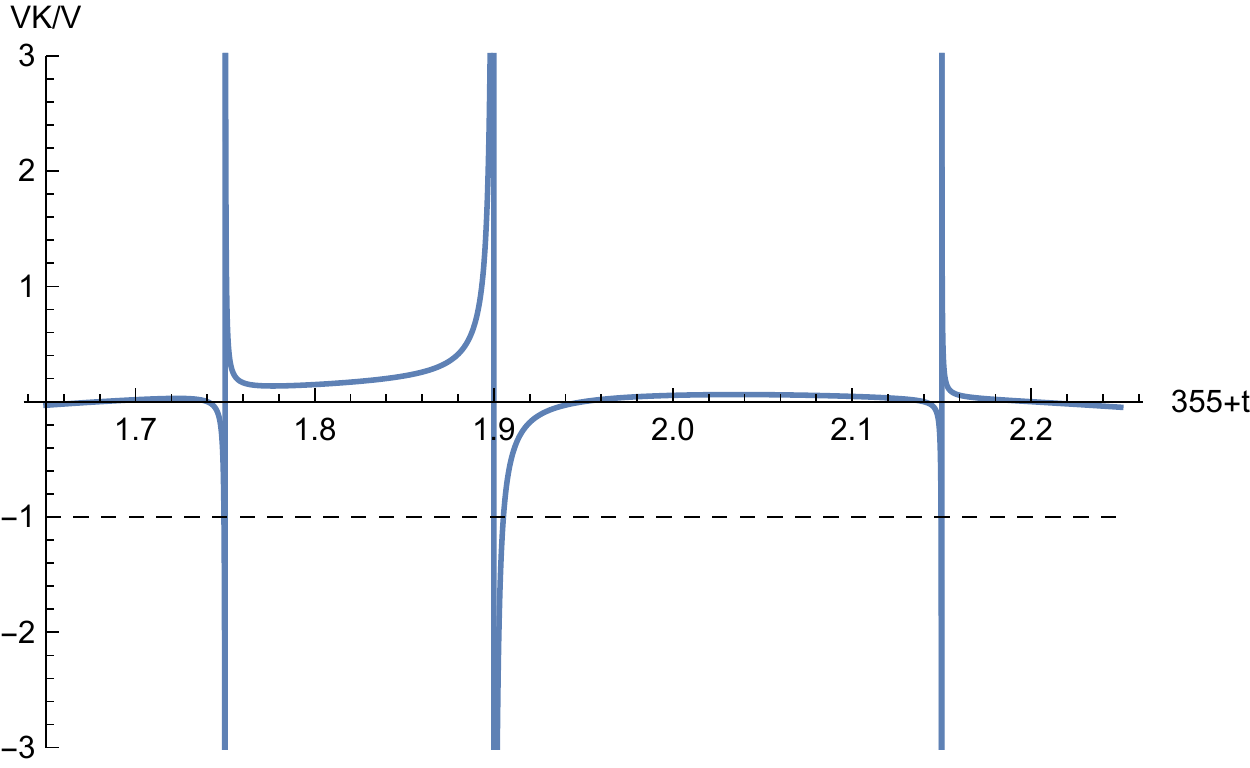}
\caption{Schematic plots of ${\cal V}_{\cal K}(1,1;1/2+i(355+t))/{\cal V}(1/2+i(355+t))$ versus $t$ , showing
the two cases beginning and ending with an enclave: at left, $ {\cal K},{\cal T}_-,{\cal L},{\cal K}_\lambda:{\cal K}, {\cal T}_+  : {\cal T}_-,{\cal L},{\cal K}_\lambda,{\cal K}$; at right,  $ {\cal K},{\cal T}_-,{\cal L},{\cal K}_\lambda :{\cal T}_+,{\cal K} : {\cal T}_-,{\cal L},{\cal K}_\lambda,{\cal K}$. The dashed line intersects the continuous curves at three points: the first and third denote zeros of ${\cal L}(s)$, and the second denotes a zero of $S_0(s)$.}
\label{figencl2}
\end{figure}

{\bf Remark:} We have thus shown for all  arrangements of zeros 
of  ${\cal K}(0,0;s)$,  ${\cal K}_\lambda (0,0;s)$, ${\cal T}_+(s)$ and ${\cal T}_-(s)$ which to our knowledge are possible  that  the zeros of $S_0(s)$  in an arbitrary island lie on the critical line.
\subsection{Comment on the Multiplicity of Zeros}
As we remarked above, it has been proved by Conrey, Iwaniec and Soundararajan\cite{conrey3}  that at least 56\% of the  non-trivial zeros in the family of all Dirichlet $L$ functions are simple and lie on the critical line. The results we have established here also enable us to comment on the question of the multiplicity of zeros of $S_0(s)$, which related not only to the multiplicity of the zeros
of $\zeta (s)$ and $L_{-4}(s)$, but also to the possibility of coincidence of zeros of these two functions.
\begin{theorem}
The only possible location for non-simple zeros of $S_0(s)$ is on the critical line, at its intersections with  boundaries of inner islands in $\sigma \ne 1/2$.
\label{thm-mult}
\end{theorem}
\begin{proof}
We know from Theorem \ref{mthm} that zeros of  $S_0(s)$ must lie on contours of unit modulus of ${\cal U}_{\cal K}(1,1;s)/{\cal U}(s)$ and ${\cal V}_{\cal K}(1,1;s)/{\cal V}(s)$. We have also stated in Remark 3 that zeros of $S_0(s)$ are simple in extended regions or enclaves, since $\arg {\cal U}_{\cal K}(1,1;s)$ and ${\cal U}(s)$ respectively increase/decrease as $t$ increases in extended regions or enclaves, with the latter decreasing monotonically. (Indeed, fromTheorem \ref{aboutL} we know further that  there are no zeros of $S_0(s)$ in enclaves.) 

We next consider the case of off-axis zeros located on the boundaries of inner islands $|{\cal U}_{\cal K}(1,1;s)/{\cal U}(s)|=1$ .
 The logarithmic potential associated with this  function  on the boundary of inner islands has a real part which is identically zero  (the Dirichlet condition) and its imaginary part in consequence obeys the condition that its normal derivative is zero (the Neumann condition). Its only singular points within the inner island are those where ${\cal U}_{\cal K}(1,1;s)/ {\cal U}(s)$ has its pole or zero. If $s_0$ denotes the position of the zero, the pole is at $1-\overline{s_0}$. A good discussion of such boundary value problems is contained in Chapter XII of the book of O.D. Kellogg \cite{kellogg}. The solution can be found if we prescribe the functional form of the cavity boundary and of the two source points. We note that the cavity is symmetrical under reflection in the critical line, and therefore the potential theory problem can be broken up into two single-source parts: that in $\sigma\ge 1/2$ and $\sigma\le 1/2$. These single-source problems then correspond to the discussion of Theorem VII \cite{kellogg}, and thus there are no points at which
 the potential gradient of the analytic potential vanishes on the cavity boundary, except the two points where it has a corner, i.e. the points where the inner island boundary in $\sigma\ne 1/2$ cuts the critical line. 
\end{proof}

{\bf Acknowledgement:} 

This paper is dedicated to the memory of the late Professor J.M. Borwein, a distinguished colleague and friend.

\clearpage
\newpage

\end{document}